\documentclass[preprint,12pt,authoryear]{elsarticle}
\usepackage[utf8]{inputenc}
\usepackage{amsmath,bm,graphicx,booktabs}
\usepackage{geometry}
\usepackage{indentfirst,float,appendix,lmodern}
\usepackage{multirow}
\usepackage[authoryear]{natbib}
\usepackage{algorithm,algorithmic}
\usepackage{hyperref}
\usepackage[english]{babel}
\usepackage{amsthm}
\usepackage{color}
\usepackage{courier}
\usepackage{amssymb}
\usepackage{mathrsfs}
\usepackage{enumitem}
\usepackage{caption}
\usepackage{subcaption}
\usepackage{comment}
\usepackage{setspace}
\usepackage{enumitem}

\usepackage{multirow}

\newcommand{\bc}{\color{black}}

\allowdisplaybreaks

\newcommand{\Cc}{\mathcal{C}}

\newcommand{\Ec}{\mathcal{E}}

\newcommand{\Gc}{\mathcal{G}}

\newcommand{\Tc}{\mathcal{T}}

\newcommand{\Vc}{\mathcal{V}}

\newtheorem{thm}{Theorem}

\newtheorem{cor}{Corollary}
\newtheorem{lem}{Lemma}

\newtheorem{rem}{Remark}

\newcommand{\vbeta}{\boldsymbol \beta}
\newcommand{\veta}{\bm{\eta}}
\newcommand{\vxi}{\bm{\xi}}
\newcommand{\vvarsigma}{\bm{\varsigma}}

 \journal{Journal}      

\begin{document}

\begin{frontmatter}

\title{
Simultaneously detecting spatiotemporal changes with penalized Poisson regression models}

\author{Zerui Zhang}

\author[3]{Xin Wang\corref{cor1}}
\ead{xwang14@sdsu.edu}

\cortext[cor1]{Corresponding author}

\author{Xin Zhang}

\author{Jing Zhang}


%
\affiliation[3]{organization={Department of Mathematics and Statistics, San Diego State University},
            addressline={5500 Campanile Drive}, 
            city={San Diego},
            postcode={92182}, 
            state={CA},
            country={USA}}


\begin{abstract}
In the realm of large-scale spatiotemporal data, abrupt changes are commonly occurring across both spatial and temporal domains. To address the concurrent challenges of detecting change points and identifying spatial clusters within spatiotemporal count data, an innovative method is introduced based on the Poisson regression model. The proposed method employs doubly fused penalization to unveil the underlying spatiotemporal change patterns. To efficiently estimate the model, an iterative shrinkage and threshold based algorithm is developed to minimize the doubly penalized likelihood function. The reliability and accuracy is confirmed by the statistical consistency properties. Furthermore, extensive numerical experiments are conducted to validate the theoretical findings, thereby highlighting the superior performance of the proposed method when compared to existing competitive approaches. 
\end{abstract}

\end{frontmatter}

\section{Introduction}
\label{sec:intro}

In disease and epidemiology studies, datasets are commonly represented in the format of spatiotemporal count data \citep{schmid2004bayesian, tzala2008bayesian}. This data format encapsulates the total number of identified cases within contiguous non-overlapping areal units over consecutive time periods, offering a comprehensive perspective on the spatial and temporal dynamics of disease prevalence \citep{ansari2020spatiotemporal}. One salient characteristic in the spatiotemporal count data is that observations from geographically proximate areal units and temporally close periods tend to exhibit more similar values than those farther apart \citep{hardisty2010analysing}. The explicit identification and quantification of fluctuating and changing patterns over both spatial and temporal domains emerge as critical research objectives. This is particularly essential in disease management, where timely insights derived from data mining facilitate the swift implementation of prevention and control measures \citep{kulldorff2001prospective, kulldorff2005space}. Beyond public health and epidemiology, the implications of spatiotemporal count data change detection reverberate across diverse domains, including agriculture \citep{besag1999bayesian, paradinas2017spatio, zhang2023simulation}, environment \citep{gusev2008temporal, lee2021clustered}, and social science \citep{law2014bayesian}.

Within spatiotemporal count data, the exploration of changes reveals two prominent dimensions: alterations in spatial relationships and shifts in the temporal domain. First, for spatial relationships, clustering has been widely studied and applied, and \cite{ansari2020spatiotemporal} gave a comprehensive review of spatiotemporal clustering approaches. \cite{anderson2017spatial} employed a Poisson regression model to identify spatial clusters in the yearly counts of respiratory admissions to hospitals. \cite{napier2019bayesian} proposed a new Bayesian approach to identify groups of areal units with similar temporal disease trends. 
\cite{siljander2022spatiotemporal} used a Poisson space–time scan statistic to detect clusters varied over both time and space in Helsinki. \cite{kulldorff2001prospective} introduced a space-time scan statistic based on Poisson likelihood, which was applied in different studies \citep{rogerson2001monitoring,guemes2021syndromic,mohammadi2022homicide}. \cite{assunccao2007space} proposed a scored based space-time scan for point processes data. \cite{frevent2021detecting} and \cite{smida2022wilcoxon} proposed distribution-free scan statistics for detecting spatial clusters in functional data. \cite{kamenetsky2022regularized} used a regularized approach to detect spatial clusters, considering time effects through a Poisson regression model. Furthermore, there are methods taking into consideration the effects of covariates in the context of spatiotemporal count data. \cite{jung2009generalized} constructed a scan statistic based on generalized linear models to adjust the effects of covariates for spatial data. \cite{lee2017cluster} and \cite{lee2020spatial} considered spatial cluster detection for regression coefficients based on hypothesis tests. \cite{lee2021clustered} extended the approach to spatiotemporal data based on a varying coefficient regression model. In recent studies, the penalization method has been adopted to discover model-based clusters through regression coefficients in spatial data \citep{li2019spatial,wang2019clustering,ma2020exploration,wang2019spatial,LIN2022105023, wang2024scanner}. 
Specifically,  these approaches incorporate the fusion penalty to have sparsity in the differences of model coefficients. A zero coefficient difference implies that two locations will have the same estimated coefficients, which indicates that they come from the same cluster. 

When delving into the changes over the time domain, temporal changes are often studied separately from spatial clustering through change point detection.
For example, the CUSUM procedure (or cumulative summation) is a well-known method for temporal change-point detection \citep{cho2015multiple,cho2016change,gromenko2017detection,wang2018high}.
The CUSUM transformation generates test statistics, which will be compared with the standard Brownian Bridge to test the existence of the change point. 
However, the traditional CUSUM method cannot be easily applied to detect temporal change points over large-scale spatial domains.
\cite{altieri2015changepoint} introduced the Bayesian approach by log-Gaussian Cox process model, and the posterior distribution of the potential change-point distribution assumed the spatial homogeneous setup. Score statistics that capture changes in the mean and/or the spatiotemporal covariance were discussed under the spatiotemporal data in \cite{xie2012spectrum}. \cite{harchaoui2010multiple} used fused lasso to detect multiple change points in time series data and was extended to autoregressive time series models 
\citep{chan2014group} and least absolute deviation based models \citep{li2020robust}. {\bc \cite{zhao2024composite} proposed an approach based on a composite likelihood to detect multiple change points in a nonstationary spatio-temporal process}. Note that none of the temporal change-point detection methods mentioned above can handle the covariate effects.

However, to the best of our knowledge, there is currently no penalization method in the literature aforementioned that can simultaneously address the tasks of change points detection and spatial clusters identification with the consideration of covariates. In this work, we propose a novel penalized approach to simultaneously address the tasks of detecting change points and identifying spatial clusters for count data. We formulate an
optimization problem based on Poisson likelihood and two penalty terms: a fused penalty for detecting change points and an adaptive spanning tree based fusion penalty for identifying spatial clusters. We develop an algorithm to solve the doubly-penalized estimation problem. We also investigate the theoretical properties of our proposed estimator. It is shown that our estimator is statistically consistent and has the capability to reveal spatial clusters
and temporal changes with probability one. Our theoretical findings are also validated by thorough numerical experiments. Besides that, covariates are also incorporated in the modeling to ensure that the detected spatial and temporal changes are robust and appropriately adjusted for relevant covariate influences.

The article is organized as follows. 
In Section \ref{Sec:model}, we will propose our model and develop the model estimation method. 
In Section \ref{Sec:theory}, we will establish the theoretical properties of our proposed estimator. In Section \ref{Sec:sim}, we will conduct simulation studies to evaluate our proposed approach under different scenarios. The proposed approach will be applied to a real example 
in Section \ref{Sec:example}.

\section{Methodology}
\label{Sec:model}

In Section \ref{subsec:model}, we introduce the statistical model and the optimization problem. And in Section \ref{subsec:algorithm}, we present the proposed algorithm to solve the optimization problem.

\subsection{Statistical Model}
\label{subsec:model}
Let $y_{it}$ be the observed count and  $n_{it}$ be the population size for location $i$ at time $t$, where $i=1,\dots, N$ and $t=1,\dots, T$. 
Consider the following Poisson regression model,
\begin{equation}
\label{eq:poisson_model}
   y_{it}\sim \text{Poisson}\left(n_{it}\mu_{it}\right), 
\end{equation}
where $\mu_{it}$ is the rate with $E(y_{it}) = n_{it}\mu_{it}$. If $y_{it}$ denotes the number of deaths for a specific cancer, then $\mu_{it}$ represents the mortality rate. We will model $\mu_{it}$ based on covariates effects and time effects. Generally, we assume two types of covariates, as in \cite{ma2020exploration} and \cite{wang2019spatial}. Specifically,  $\bm{z}_{it}$ is the covariate vector with dimension $q$, which has common effects across all locations, and $\bm{x}_{it}$ is the covariate vector with dimension $p$, which has location-specific effects. Assume the following model for $\mu_{it}$,
\begin{equation}
\label{eq:rate}
\log\mu_{it}= \bm{z}_{it}^{\top}\bm{\alpha} + \bm{x}_{it}^{\top}\bm{\beta}_{i}+\eta_{t},
\end{equation}
where $\bm{\alpha}$ represents the vector of common regression coefficients shared by global effects, $\bm{\beta}_i$'s are location-specific regression coefficients, and $\eta_{t}$ is the time effect. A special case is that  $\bm{x}_{it} =1$, which represents the location-specific intercept. Then, the model will have a simplified form,
\begin{equation}
\label{eq:rate_intercept}
    \log\mu_{it}= \bm{z}_{it}^{\top}\bm{\alpha} + \beta_i+\eta_{t}.
\end{equation}

Assume that $N$ locations are from $K$ underlying spatial clusters $\{\Cc_k\}_{k=1}^K$, where $\mathcal{C}_k$ contains locations belonging to cluster $k$. That is $\vbeta_i = \vbeta_{i^\prime}$ if location $i$ and location $i^{\prime}$ are both in cluster $k$, for $k = 1,\dots, K$. Furthermore, assume that there are $J$ change points $t_j^*$ such that $\eta_t = \tau_{j}$ if $t^*_{j-1} \leq t \leq t^*_j - 1$ for $j=1,\dots, J+1$ with $t^*_0 = 1$ and $t^*_{J+1} = T+1$. For an identifiability purpose, we assume that $\eta_1 = 0$, indicating that $\tau_1 = 0$. Our goal is to use observed data to estimate the number of clusters $\hat{K}$, the corresponding cluster structure $\hat{\mathcal{C}} = \{\hat{\mathcal{C}}_1,\dots, \hat{\mathcal{C}}_{\hat{K}}\}$, the estimated regression coefficients $\hat{\bm{\alpha}}$, $\hat{\bm{\beta}}_i$, the number of changed points $\hat{J}$, and the corresponding time effects $\hat{\eta}_t$.

To achieve the goal, we will construct an optimization problem based on the following likelihood function and two sets of penalty functions. Let $\bm{\eta} = (\eta_1,\eta_2,\dots,\eta_T)^\top$ and $\bm{\beta} = (\bm{\beta}^\top_1,\dots,\bm{\beta}^\top_N)^\top$, the negative loglikelihood function based on the model in \eqref{eq:poisson_model} and \eqref{eq:rate} is
\begin{equation}
\label{Eq:nll}
   l_0\left(\bm{\alpha},\bm{\eta},\bm{\beta}\right)=\sum_{i=1}^{N}\sum_{t=1}^{T}\left[-y_{it}\left(\log n_{it}+\bm{z}_{it}^{\top}\bm{\alpha} + \bm{x}_{it}^{\top}\bm{\beta}_{i}+\eta_{t}\right)+n_{it}\exp\left(\bm{z}_{it}^{\top}\bm{\alpha} + \bm{x}_{it}^{\top}\bm{\beta}_{i}+\eta_{t}\right)\right].
\end{equation}

To find change points, we consider a fused type penalty on time effects as proposed in \cite{harchaoui2010multiple}, which has the following form,
\begin{equation}
    \label{Eq:pen1}
     \sum_{t=2}^{T}\mathcal{P}_{\gamma_{1}}\left(\vert\eta_{t}-\eta_{t-1}\vert,\lambda_{1}\right), 
\end{equation}
{\bc where $\mathcal{P}_{\gamma_{1}}(\cdot,\lambda_1)$ is a penalty function. There are different penalty functions in the literature, such as $L_1$ \citep{tibshirani1996regression}, the adaptive lasso \citep{zou2006adaptive}, the smoothly clipped absolute deviation (SCAD) penalty \citep{fan2001variable} and the minimax concave penalty (MCP)\citep{zhang2010nearly}. \cite{ma2017concave} compared different penalty functions in terms of clustering performance under the linear regression setups and found that SCAD and MCP perform similarly, and $L_1$ tended to give more clusters.  In our work, we use MCP. Note that the algorithm can be easily adapted to the SCAD penalty and the theoretical properties also hold for the SCAD penalty.} In particular, the MCP has the following form
\begin{align}
\mathcal{P}_{\gamma}(t,\lambda)=\begin{cases}
\lambda\vert t\vert-\frac{t^{2}}{2\gamma}, & \vert t\vert\leq\gamma\lambda,\\
\frac{1}{2}\gamma\lambda^{2}, & \vert t\vert>\gamma\lambda.
\end{cases}
\end{align}
$\gamma$ is a built-in parameter, which is fixed at 3 as in the literature \citep{ma2020exploration}, $\lambda$ is a tuning parameter, which will be selected based on data driven criteria. 

If $\vert \eta_t - \eta_{t-1}\vert$ is shrunk to zero, no change point is identified at time $t$. If  $\vert \eta_t - \eta_{t-1}\vert$ is not shrunk to zero, then a change point at time $t$ is detected. Let $\bm{\xi}= (\xi_2,\dots, \xi_T)^{\top} = \left(\eta_{2} -\eta_1,\eta_{3}-\eta_{2},\dots,\eta_{T}-\eta_{T-1}\right)^{\top}$, then, $(\eta_2,\dots, \eta_T)^\top  = \bm{M}\bm{\xi}$, where $\bm{M}$ is a $(T-1)\times (T-1)$ lower triangular matrix with nonzero elements equal to one. 
Thus \eqref{Eq:pen1} can be written as, 
\begin{equation}
\label{Eq:pen1new}
    \sum_{t=2}^{T}\mathcal{P}_{\gamma_{1}}\left(\vert\eta_{t}-\eta_{t-1}\vert,\lambda_{1}\right) = \sum_{t=2}^{T} \mathcal{P}_{\gamma_{1}}\left(\vert\xi_t \vert,\lambda_{1}\right). 
\end{equation}
\eqref{Eq:pen1new} implies that detecting zero differences between $\eta_t$ and $\eta_{t-1}$ is equivalent to detecting zero values of $\xi_t$.

To find the spatial cluster pattern, we will consider a graph based fusion penalty. Let $\Gc$ be an undirected connected graph $\Gc = (\Vc,\Ec_0),$ where $\Vc = \{v_1,\cdots,v_N\}$ is the set of vertices with $v_i$ representing location $i,$ and $\Ec_0 = \{(v_i,v_{i^{\prime}}):v_i\neq v_{i^{\prime}}\}$ is the edge set. In areal data, we can construct this graph based on neighbor structure: if area $i$ and area $i^\prime$ share a boundary, then $(v_i,v_{i^\prime}) \in \Ec_0$. In geostatistical data, we can construct this graph based on Delaunay triangulation \citep{lee1980two}. A spanning tree $\Tc$ of the graph $\Gc$ is a connected undirected subgraph of $\Gc$ with no cycles and includes all the vertices of $\Gc$. A special spanning tree is a minimum spanning tree (MST). Denote $d(v_i, v_{i^\prime})$ as the associated weight to each edge $(v_i, v_{i^\prime})$ in $\mathcal{E}_0$, then a MST is defined as $\mathcal{T} = (\mathcal{V}, \mathcal{E})$ such that this subgraph is a spanning tree and $\sum_{(v_i, v_{i^{\prime}})\in \mathcal{E}}d(v_i, v_{i^{\prime}})$ is minimized \citep{li2019spatial}.
The construction of MST relies on the edge weights for all edges in $\Ec$. If there is some prior knowledge of the weights, the weights can be constructed based on local estimates as used in \cite{zhang2019distributed}, \cite{wangZhangZhu2023} and \cite{wang2024scanner}. If there are no local estimates, one can use distance to define weights for geographical data as used in \cite{li2019spatial}. If the weights are not properly defined,  MST will be constructed based on equal weights.  The tree based penalty for a given an MST, $\mathcal{T}$, is defined as
\begin{equation}
\label{Eq:pen2}
    \sum_{(i,i^{\prime})\in\mathcal{E}}\mathcal{P}_{\gamma_{2}}\left(\Vert\bm{\beta}_{i}-\bm{\beta}_{i^{\prime}}\Vert,\lambda_{2}\right),
\end{equation}
where {\bc $\Vert \cdot \Vert$ is the Euclidean norm}$, \mathcal{P}_{\gamma_2}(\cdot,\lambda_2)$ is the MCP with $\gamma_2 = 3$ and $\lambda_2$ is a tuning parameter that will be selected later.

Based on the property of MST, we can establish an incident matrix $\bm{H}$ for $\mathcal{T}$, which is an $(N-1)\times N$ full rank matrix. And the $(l,i)$th entry in $\bm{H}$ is defined as: $\bm{H}_{l,i}=\begin{cases}
1, & \text{if }i=s(l),\\
-1, & \text{if }i=e(l),\\
0, & \text{otherwise},
\end{cases}$ where $s(l)$ and $e(l)$ denote the starting and ending node indices of edge $l$ in $\mathcal{T}$, respectively, with $s(l) < e(l)$.  Then, $
    \left(\bm{\beta}_{i}^{\top}-\bm{\beta}_{i^{\prime}}^{\top},\left(i,i^{\prime}\right)\in\mathcal{T}\right)=\left(\bm{H}\otimes\bm{I}_{p}\right)\bm{\beta}$,  {\bc where $\otimes$ represents the Kronecker product}. Let $\bm{\tilde{H}}=\left(\begin{array}{c}
\frac{1}{\sqrt{N}}\bm{1}^{\top}\\
\bm{H}
\end{array}\right)$, which is an $N\times N$ full rank matrix.  Then, we can define $\bm{\varsigma}= (\vvarsigma_1^{\top},\vvarsigma_2^{\top},\dots, \vvarsigma_N^{\top})^{\top} = (\tilde{\bm{H}}\otimes \bm{I}_p)\bm{\beta}$. And $\vbeta = (\tilde{\bm{H}}\otimes \bm{I}_p)^{-1}\vvarsigma = (\tilde{\bm{H}}^{-1}\otimes \bm{I}_p)\vvarsigma$. Thus, \eqref{Eq:pen2} can be written as
\begin{equation}
\label{Eq:pen2new}
    \sum_{(i,i^{\prime})\in\mathcal{E}}\mathcal{P}_{\gamma_{2}}\left(\Vert\bm{\beta}_{i}-\bm{\beta}_{i^{\prime}}\Vert,\lambda_{2}\right) = \sum_{i=2}^{N}\mathcal{P}_{\gamma_{2}}\left(\Vert\bm{\vvarsigma}_i\Vert,\lambda_{2}\right).
\end{equation}
And \eqref{Eq:pen2new} implies that detecting zero differences between $\bm{\beta}_i$ and $\bm{\beta}_{i^\prime}$ is equivalent to detecting zero values of $\bm{\vvarsigma}_i$.

To achieve the goal of identifying change points and finding cluster patterns simultaneously, we consider minimizing the following objective function, which combines the likelihood function in \eqref{Eq:nll}, the penalty functions for identifying change points in \eqref{Eq:pen1} and the penalty functions for finding clusters in \eqref{Eq:pen2}, 
\begin{align}
\label{Eq:obj}
L(\bm{\alpha},\bm{\eta},\bm{\beta}) =  & \frac{1}{NT}\sum_{i=1}^{N}\sum_{t=1}^{T}\left[-y_{it}\left( \bm{z}_{it}^{\top}\bm{\alpha} + \bm{x}_{it}^{\top}\bm{\beta}_{i}+\eta_{t}\right)+n_{it}\exp\left(\bm{z}_{it}^{\top}\bm{\alpha} + \bm{x}_{it}^{\top}\bm{\beta}_{i}+\eta_{t}\right)\right]\nonumber\\
+ & \sum_{t=2}^{T}\mathcal{P}_{\gamma_{1}}\left(\vert\eta_{t}-\eta_{t-1}\vert,\lambda_{1}\right)+\sum_{(i,i^{\prime})\in\mathcal{E}}\mathcal{P}_{\gamma_{2}}\left(\Vert\bm{\beta}_{i}-\bm{\beta}_{i^{\prime}}\Vert,\lambda_{2}\right).
\end{align}

Let $\bm{\theta}=(\bm{\alpha}^\top, \vxi^{\top},\vvarsigma^{\top})^{\top}$, which is a $(q+T-1+Np)\times 1$ vector. Based on \eqref{Eq:pen1new} and \eqref{Eq:pen2new}, the objective function \eqref{Eq:obj} can be expressed in terms of $\bm{\theta}$ as follows:
\begin{align}
\label{Eq:obj_new}
Q(\bm{\theta}) =  & \frac{1}{NT}\sum_{i=1}^{N} \sum_{t=1}^{T}\left[-y_{it}\left(\bm{z}_{it}^{\top}\bm{\alpha} + \bm{x}_{it}^{\top}\bm{h}_{i}\vvarsigma+\bm{m}_t^{\top}\vxi\right)+n_{it}\exp\left(\bm{z}_{it}^{\top}\bm{\alpha} + \bm{x}_{it}^{\top}\bm{h}_{i}\vvarsigma+\bm{m}_t^{\top}\vxi\right)\right]\nonumber\\
+ & 
\sum_{t=2}^{T}\mathcal{P}_{\gamma_{1}}\left(\vert\xi_t\vert,\lambda_{1}\right)+\sum_{i=2}^{N}\mathcal{P}_{\gamma_{2}}\left(\Vert\bm{\vvarsigma}_i\Vert,\lambda_{2}\right),
\end{align}
where $\bm{m}_t^{\top}$ is the $(t-1)$th row of $\bm{M}$ for $t=2,\dots, T$, $\bm{m}_1 = \bm{0}_{T-1}$, a $(T-1)\times 1 $ zero vector, and $\bm{h}_i = \tilde{\bm{h}}_i^{\top} \otimes \bm{I}_p$ with $\tilde{\bm{h}}_i^{\top}$ is the $i$th row of $\tilde{\bm{H}}^{-1}$.

\subsection{Computation algorithm}
\label{subsec:algorithm}

Let $\hat{\bm{\theta}}=(\hat{\bm{\alpha}}^\top, \hat{\vxi}^{\top},\hat{\vvarsigma}^{\top})^{\top}$ be the solution to the following minimization problem for a given tree $\Tc$ and given $\lambda_1$ and $\lambda_2$ based on the objective function in \eqref{Eq:obj_new},
\begin{equation}
    \label{Eq:sol}
    \hat{\bm{\theta}} = \arg \min_{\bm{\theta}}Q(\bm{\theta};\lambda_1,\lambda_2).
\end{equation}

To solve the minimization problem, we develop an algorithm based on the general iterative shrinkage and thresholding algorithm (GIST)\citep{gong2013general}. The details are outlined below.

Denote $l(\bm{\theta}) = \frac{1}{NT}\sum_{i=1}^{N} \sum_{t=1}^{T}\left[-y_{it}\left(\bm{z}_{it}^{\top}\bm{\alpha} + \bm{x}_{it}^{\top}\bm{h}_{i}\vvarsigma+\bm{m}_t^{\top}{\vxi}\right)+n_{it}\exp\left(\bm{z}_{it}^{\top}\bm{\alpha} + \bm{x}_{it}^{\top}\bm{h}_{i}\vvarsigma+\bm{m}_t^{\top}{\vxi}\right)\right]$. Given the current values of parameters $\bm{\theta}^{(r)}$ at the $r$th step, then the $(r+1)$th update $\bm{\theta}^{(r+1)}$ is given by,
{\bc 
\begin{align*}
\bm{\theta}^{(r+1)}= & \underset{\bm{\theta}}{\text{argmin}}\ l\left(\bm{\theta}^{(r)}\right)+\left\langle \nabla l\left(\bm{\theta}^{\left(r\right)}\right),\bm{\theta}-\bm{\theta}^{\left(r\right)}\right\rangle +\frac{\rho^{\left(r\right)}}{2}\Vert\bm{\theta}-\bm{\theta}^{(r)}\Vert^{2}\\
 & +\sum_{t=2}^{T}\mathcal{P}_{\gamma_{1}}\left(\vert\xi_{t}\vert,\lambda_{1}\right)+\sum_{i=2}^{N}\mathcal{P}_{\gamma_{2}}\left(\Vert\bm{\vvarsigma}_{i}\Vert,\lambda_{2}\right),
\end{align*}}
where $\nabla l\left(\bm{\theta}^{\left(r\right)}\right)$ is the first order derivative of $l(\cdot)$, $ \left\langle \cdot, \cdot \right \rangle$ is the inner product and $1/\rho^{(r)}$ is the step size. The step size is determined by the line search criterion used in \cite{gong2013general}. The detail is provided in Remark \ref{rem1}. Following the GIST algorithm, the problem is equivalent to the following proximal operator problem:
\begin{equation}
\label{Eq:sol_update}
   \bm{\theta}^{(r+1)}=\underset{\bm{\theta}}{\text{argmin}}\frac{1}{2}\Vert\bm{\theta}-\bm{u}^{(r)}\Vert^{2}+\frac{1}{\rho^{(r)}}\sum_{t=2}^{T}\mathcal{P}_{\gamma_{1}}\left(\vert\xi_{t}\vert,\lambda_{1}\right)+\frac{1}{\rho^{(r)}}\sum_{i=2}^{N}\mathcal{P}_{\gamma_{2}}\left(\Vert\bm{\vvarsigma}_{i}\Vert,\lambda_{2}\right),
\end{equation}
where $\bm{u}^{(r)}=\bm{\theta}^{(r)}-\nabla l\left(\bm{\theta}^{(r)}\right)/\rho^{(r)}$, and {\bc $\xi_t$ and $\bm{\varsigma}_i$ are part of $\bm{\theta}$, which need to be updated}.  Since there are no penalties applied to $\bm{\alpha}$ and $\bm{\varsigma}_1$, $\bm{\alpha}^{(r+1)}$ and $\bm{\varsigma}_1^{(r+1)}$ are updated based on the values in $\bm{u}^r$ directly. In particular, $\bm{\alpha}^{(r+1)} = \bm{u}^{(r)}_{[1:q]}$, where $\bm{u}^{(r)}_{[1:q]}$ is first $q$ values in $\bm{u}^{(r)}$ correspond to $\bm{\alpha}$. $\bm{\varsigma}_{1}^{(r+1)}=\bm{u}_{[(q+T):(q+T+p-1)]}^{(r)}$
, where $\bm{u}_{[(q+T):(q+T+p-1)]}^{(r)}$ is the $q+T$ to $q+T+p-1$ elements,
corresponding to $\bm{\varsigma}_{1}$ in $\bm{\theta}$.

To update $\xi_{t}^{(r+1)}$, $t=2,\dots T$, it is equivalent to the following minimization problem, 
\[
\xi_{t}^{(r+1)}=\underset{\xi_{t}}{\text{argmin}}\frac{\rho^{(r)}}{2}\left(\xi_{t}-u_{[q+t-1]}^{(r)}\right)^{2}+\mathcal{P}_{\gamma_{1}}\left(\vert\xi_{t}\vert,\lambda_{1}\right),
\]
where $u_{[q+t-1]}^{(r)}$ the $(q+t-1)$the element in $\bm{u}^{(r)}$. To update $\bm{\varsigma}_{i}^{(r+1)}$, $i=2,\dots N$, it is equivalent to the following minimization problem,
\[
\bm{\varsigma}_{i}^{(r+1)}=\underset{\bm{\varsigma}_{i}}{\text{argmin}}\frac{\rho^{(r)}}{2}\Vert\bm{\varsigma}_{i}-\bm{u}_{\bm{\varsigma}_{i}}^{(r)}\Vert^{2}+\mathcal{P}_{\gamma_{2}}\left(\Vert\bm{\vvarsigma}_{i}\Vert,\lambda_{2}\right),
\]
where $\bm{u}_{\bm{\varsigma}_{i}}^{(r)}$ is the value in $\bm{u}^{(r)}$ corresponding to $\bm{\varsigma}_i$.
For MCP \citep{zhang2010nearly}, the solutions for $\xi_{t}$ and $\bm{\varsigma}_i$ are
\begin{equation}
\label{Eq:sol_xi}
\xi_{t}^{(r+1)}=\begin{cases}
\frac{S\left(u_{[q+t-1]}^{(r)},\lambda_{1}/\rho^{(r)}\right)}{1-1/\left(\gamma_{1}\rho^{(r)}\right)} & \text{if }\vert u_{[q+t-1]}^{(r)}\vert\leq\gamma_1\lambda_1,\\
u_{[q+t-1]}^{(r)} & \text{if }\vert u_{[q+t-1]}^{(r)}\vert>\gamma_1\lambda_1,
\end{cases}
\end{equation}
and 
\begin{equation}
\label{Eq:sol_varsigma}
\bm{\varsigma}_{i}^{(r+1)}=\begin{cases}
\frac{S\left(\bm{u}_{\bm{\varsigma}_{i}}^{(r)},\lambda_{2}/\rho^{(r)}\right)}{1-1/\left(\gamma_{2}\rho^{(r)}\right)} & \text{if }\Vert\bm{u}_{\bm{\varsigma}_{i}}^{(r)}\Vert\leq\gamma_2\lambda_2,\\
\bm{u}_{\bm{\varsigma}_{i}}^{(r)} & \text{if }\Vert\bm{u}_{\bm{\varsigma}_{i}}^{(r)}\Vert>\gamma_2\lambda_2,
\end{cases}
\end{equation}
where $S(\bm{x},\lambda) = (1-\lambda/\Vert \bm{x}\Vert)_+ \bm{x}$, and $(x)_+=x$ if $x>0$, 0 otherwise.

The estimates $\hat{\bm{\varsigma}}$  highly depend on the pre-selected tree $\Tc$, which has been discussed in \cite{li2019spatial}, \cite{zhang2019distributed} and \cite{LIN2022105023}. Figure \ref{fig:mst_examples} illustrates two examples of spanning trees. In these two figures, there are two clusters within 25 nodes (locations). In Figure \ref{fig:mst1}, with a properly constructed tree, there is one edge between two clusters. Then, the graph can be partitioned into two subgraphs by removing the edge between them (the red cross in the figure), representing two clusters. However, in Figure \ref{fig:mst2}, if the tree is not properly constructed, there are two edges between these two clusters. By removing these two edges (red crosses in the figure), the graph is partitioned into three subgraphs, leading to three clusters. 

\begin{figure}[H]
\centering
\begin{subfigure}{.45\textwidth}
  \centering
  \includegraphics[width=0.95\linewidth]{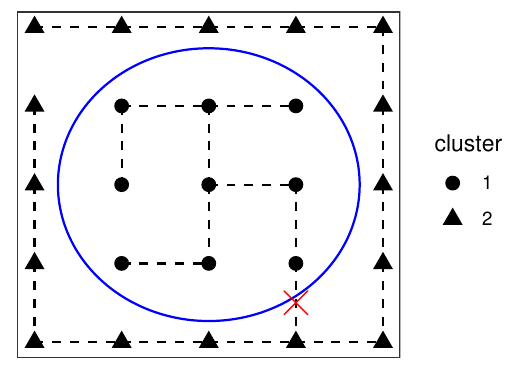}
  \caption{An example of a spanning tree with one edge between two clusters. }
  \label{fig:mst1}
\end{subfigure}%
\hspace{.1in}
\begin{subfigure}{.45\textwidth}
  \centering
  \includegraphics[width=0.95\linewidth]{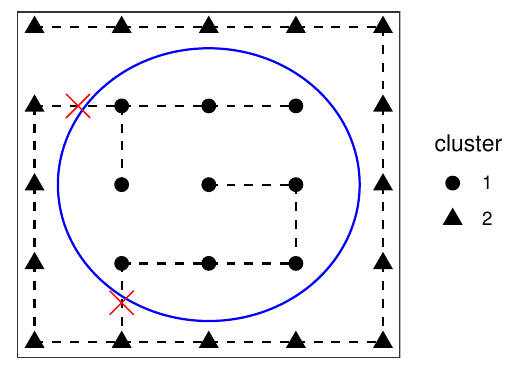}
  \caption{An example of a spanning tree with two edges between two clusters.}
  \label{fig:mst2}
\end{subfigure}
\caption{Two spanning trees with 25 nodes. Different shapes represent two different clusters. Dashed lines represent edges connecting two nodes. }
\label{fig:mst_examples}
\end{figure}
Here, we propose an adaptive approach to improve the constructed MST, which is not studied in the literature. Consider an initial MST based on neighbors or locations, $\mathcal{T} = (\mathcal{V}, \mathcal{E})$.  $\hat{\bm{\theta}}_{\mathcal{T}}$ is the estimator obtained from minimizing \eqref{Eq:obj} based on the given tree $\mathcal{T}$. According to the relationship between $\vbeta$ and $\vvarsigma$, we can obtain the estimator $\hat{\vbeta}_{\mathcal{T}} = \left(\tilde{\bm{H}}^{-1}\otimes \bm{I}_p\right)\hat{\vvarsigma}_{\mathcal{T}}$.
For any two locations in the original graph $\Gc$, we define the following weights based on estimates $\hat{\bm{\beta}}_{\mathcal{T}}$,
\begin{equation}
\label{Eq:weight}
 w_{ii^\prime} = \begin{cases}
\Vert\hat{\vbeta}_{\mathcal{T},i}-\hat{\vbeta}_{\mathcal{T},i^\prime}\Vert, & \text{if }(i,i^\prime)\in\mathcal{E}_0,\\
\infty, & \text{otherwise.}
\end{cases}   
\end{equation}
Then, a new MST $\Tc^*  = (\mathcal{V},\mathcal{E}^*)$ is constructed using the updated weights in \eqref{Eq:weight}. In the end, the final estimator $\hat{\bm{\theta}}^* = (\hat{\bm{\alpha}}^*,\hat{\vxi}^*,\hat{\vvarsigma}^*)$ is obtained by solving a minimization problem similar to \eqref{Eq:sol} based on the new MST, $\mathcal{T}^*$. In particular, the penalized likelihood function has the following form
\begin{align}
\label{Eq:obj_final}
L(\bm{\alpha},\bm{\eta},\bm{\beta}) =  & \frac{1}{NT}\sum_{i=1}^{N}\sum_{t=1}^{T}\left[-y_{it}\left( \bm{z}_{it}^{\top}\bm{\alpha} + \bm{x}_{it}^{\top}\bm{\beta}_{i}+\eta_{t}\right)+n_{it}\exp\left(\bm{z}_{it}^{\top}\bm{\alpha} + \bm{x}_{it}^{\top}\bm{\beta}_{i}+\eta_{t}\right)\right]\nonumber\\
+ & \sum_{t=2}^{T}\mathcal{P}_{\gamma_{1}}\left(\vert\eta_{t}-\eta_{t-1}\vert,\lambda_{1}\right)+\sum_{(i,i^{\prime})\in\mathcal{E}^*}\mathcal{P}_{\gamma_{2}}\left(\Vert\bm{\beta}_{i}-\bm{\beta}_{i^{\prime}}\Vert,\lambda_{2}\right),
\end{align}
where the fusion type penalty is applied to the edges in $\mathcal{E}^*$. {\bc Note that, compared to \eqref{Eq:obj}, the difference is that the penalty function is imposed on different edges, which are defined based on different trees.}
Based on the same arguments, the optimization problem is equivalent to minimizing the following objective function,
\begin{align}
\label{Eq:obj_new_final}
Q^*(\bm{\theta}) =  & \frac{1}{NT}\sum_{i=1}^{N} \sum_{t=1}^{T}\left[-y_{it}\left(\bm{z}_{it}^{\top}\bm{\alpha} + \bm{x}_{it}^{\top}\bm{h}_{i}^*\vvarsigma+\bm{m}_t^{\top}{\vxi}\right)+n_{it}\exp\left(\bm{z}_{it}^{\top}\bm{\alpha} + \bm{x}_{it}^{\top}\bm{h}_{i}^*\vvarsigma+\bm{m}_t^{\top}{\vxi}\right)\right]\nonumber\\
+ & 
\sum_{t=2}^{T}\mathcal{P}_{\gamma_{1}}\left(\vert\xi_t\vert,\lambda_{1}\right)+\sum_{i=2}^{N}\mathcal{P}_{\gamma_{2}}\left(\Vert\bm{\vvarsigma}_i\Vert,\lambda_{2}\right),
\end{align}
where $\bm{h}^*_i = \tilde{\bm{h}}^{*\top}_i \otimes \bm{I}_p$, and $\tilde{\bm{h}}^{*\top}_{i}$ is the $i$th row of $(\tilde{\bm{H}}^*)^{-1}$. Here, $\bm{\tilde{H}}^*=\left(\begin{array}{c}
\frac{1}{\sqrt{N}}\bm{1}^{\top}\\
\bm{H}^*
\end{array}\right)$, and $\bm{H}^*$ is the incident matrix corresponding to $\mathcal{T}^*$. With the estimator $(\hat{\bm{\alpha}}^*,\hat{\vxi}^*, \hat{\vvarsigma}^*) = \arg 
\min_{\bm{\theta}} Q^*(\bm{\theta})$, the estimated time effect and regression coefficients are $\hat{\veta}^* = \bm{M}\hat{\vxi}^*$ and $\hat{\vbeta}^* = \left(\tilde{\bm{H}}^* \otimes \bm{I}_p\right)^{-1}\hat{\vvarsigma}^* = \left(\tilde{\bm{H}}^{*-1} \otimes \bm{I}_p\right)\hat{\vvarsigma}^*$.

 In summary, the computational algorithm for  a given $\lambda_1$ and a given $\lambda_2$ is outlined as follows:
\begin{algorithm}[H]
	\caption{ Adaptive algorithm}
 \label{alg:ada}
	\begin{algorithmic}[1]
	    \REQUIRE: Initialize $\mathcal{T}$ and $\bm{\theta}^{(0)}$. 
	    \FOR {$r=0,1,2,\dots$}
		\STATE Initialize $\rho^{(r)} = 1$.
		      \REPEAT 
		      		   \STATE Update $\bm{\theta} ^{(r+1)}$ based on \eqref{Eq:sol_update}, \eqref{Eq:sol_xi} and \eqref{Eq:sol_varsigma}.
		      		   \STATE $\rho^{(r+1)} = 2\rho^{(r)}$.
		      \UNTIL some line search criterion is satisfied.
		\IF{convergence criterion is met} 
		\STATE{Stop and get the estimates $\hat{\bm{\theta}}_{\mathcal{T}}$.}
			\ELSE
		\STATE { $r=r+1$}
		\ENDIF
		\ENDFOR
        \STATE Obtain the updated MST, $\mathcal{T}^*$, based on $\hat{\bm{\theta}}_{\mathcal{T}}$ .
        \STATE Repeat steps 1 - 12 to obtain the final estimates based on $\mathcal{T}^*$.
	\end{algorithmic}
\end{algorithm}

\begin{rem}
\label{rem1}
\cite{gong2013general} discussed the linear search criterion in detail. Here we follow their procedure, that is $Q(\bm{\theta}^{(r+1)}) \leq Q(\bm{\theta}^{(r)}) - \frac{\sigma}{2} \rho^{(r)}\Vert \bm{\theta}^{(r+1)}  - \bm{\theta}^{(r)}\Vert $, where $\sigma = 10^{-5}$ as used in \cite{gong2013general}.  And the convergence criterion we use here is $\max \{ \vert\theta_l^{(r+1)} - \theta_l^{(r)} \vert \} < 10^{-4}$. 
\end{rem}

{\bc
\begin{rem}
    In each iteration of Algorithm 1, it conducts the updates of $\bm{\theta}$ based on  \eqref{Eq:sol_update}-\eqref{Eq:sol_varsigma}. Note that the updating in \eqref{Eq:sol_update}-\eqref{Eq:sol_varsigma} is element-wise and the dimension of $\bm{\theta}$ is $(q+T-1+Np)$. Thus, the per-iteration computation complexity is $O(T+N)$. Following Theorem 2 in \cite{gong2013general}, it would requires $O(1/\sqrt{\epsilon})$ iterations to meet the convergence criterion $\max\{|\theta^{r+1}_l - \theta^{r}_l|\}\le \epsilon$ ($\epsilon = 10^{-4}$ in our case). Thus, the total computation complexity for the iterative shrinkage and thresholding algorithm is $O((T+N)/\sqrt{\epsilon})$
\end{rem}
}

\begin{rem}
The initial values can be obtained by fitting a Poisson regression model with location-specific effects $\bm{\beta}_i$ and individual time effect $\eta_t$. 
\end{rem}

Recall that by changing the values of $\lambda_1$ and $\lambda_2$, we can obtain the estimates of the number of changed points $\hat{J}$ and the number of clusters $\hat{K}$. Following the literature \cite{ma2020exploration}, $\lambda_1$ and $\lambda_2$ are selected based on minimizing the following modified BIC,
\begin{equation}
    \label{eq:bic1}
    BIC(\lambda_1,\lambda_2) = 2l_0(\hat{\bm{\alpha}},\hat{\bm{\eta}},\hat{\bm{\beta}}) + C_N\log(NT) (\hat{K}p + \hat{J}),
\end{equation}
where $l(\cdot)$ is the negative loglikelihood in \eqref{Eq:nll}, $\hat{J}$ is the estimated number of changed points, $\hat{K}$ is the estimated number of clusters, $p$ is the dimension of $\bm{x}$, and $C_N  = \log(Np + T-1)$.

A two-step procedure is implemented to select these two tuning parameters {\bc similar to \cite{tang2023multivariate}}. First, we set $\lambda_2 = 0$ and select $\lambda_1$ for a given tree $\mathcal{T}$ based on BIC in \eqref{eq:bic1}. Specifically, $\lambda_1$ will be chosen from a grid of values of $\lambda_1$. Second, given the selected $\lambda_1$, for each $\lambda_2$, we implement the algorithm described in Algorithm \ref{alg:ada}. Then, $\lambda_2$ will be selected based on the BIC in \eqref{eq:bic1}. Thus, we can have the final estimates.

\section{Theoretical properties}
\label{Sec:theory}

In this section, we study the theoretical properties of our proposed estimator. We will study the properties in two steps. In the first step, we will present the theoretical properties for a given tree $\mathcal{T}$. In the second step, we will provide the properties for the estimator based on the adaptive tree $\mathcal{T}^*$.

First, we will introduce some notations. Let $\bm{\alpha}^{0},\bm{\eta}^{0},\bm{\beta}^{0}$ be the true parameters
vectors, where $\bm{\beta}^{0}=\left(\bm{\beta}_{1}^{0},\dots,\bm{\beta}_{N}^{0}\right)$,
and the true cluster structure is $\mathcal{C}^{0}=\left\{ \mathcal{C}_{k}^{0}\right\} _{k=1}^{K}$, where $K$ is the number of clusters.
For a given MST, $\mathcal{T}=\left(\mathbb{\mathcal{V}},\mathbb{\mathcal{E}}\right)$,
define $\mathbb{C}_{\mathcal{T}}=\left\{ \left(i,i^{\prime}\right)\in\mathcal{E}:\bm{\beta}_{i}^{0}\neq\bm{\beta}_{i^{\prime}}^{0}\right\} $, which corresponds to the edges in $\mathcal{T}$ and the nodes (locations) are not in the same cluster. {\bc For example, $\mathbb{C}_{\mathcal{T}}$ defined based on Figure \ref{fig:mst1} will contains one edge, and $\mathbb{C}_{\mathcal{T}}$ defined based on Figure \ref{fig:mst2} will contains two edges.} Then by removing $\vert\mathbb{C}_{\mathcal{T}}\vert$ number of edges,
we can obtain $K_{\mathcal{T}} \equiv \vert\mathbb{C}_{\mathcal{T}}\vert+1$ number of clusters, and the corresponding cluster structure is denoted as $\mathcal{C}_{\mathcal{T}}$. Note that, $K_{\mathcal{T}}$ may not be the same as $K$ as discussed in Figure \ref{fig:mst_examples}. Let $\bm{\theta}_{\mathcal{T}} = (\bm{\alpha}^\top,\bm{\xi}^\top,\bm{\varsigma}^\top)^\top$ be the transformed parameter vector depending on $\mathcal{T}$. Recall that $\bm{\varsigma}_l$ represents $\bm{\beta}_i - \bm{\beta}_i^\prime$ for $(i,i^\prime) \in \mathcal{E}$, where $ \mathcal{E}$ is the edge set for the tree $\mathcal{T}$. If the nonzero $\bm{\varsigma}_l$ can be identified, then the corresponding cluster structure $\mathcal{C}_{\mathcal{T}}$ can then be identified. As discussed above,  if MST is properly constructed, such that by removing $\vert \mathbb{C}_{\mathcal{T}}\vert$ edges, we can recover the true cluster structure $\mathcal{C}^0$. Because of the one-to-one mapping relationship between $\bm{\varsigma}$ and $\bm{\beta}$, $\bm{\xi}$ and $\bm{\eta}$, thus estimating $\bm{\varsigma}$ and $\bm{\xi}$ is equivalent to estimating $\bm{\beta}$ and $\bm{\eta}$. Thus, we will focus on the properties of $\hat{\bm{\theta}}$, the estimator of $\bm{\theta}=(\bm{\alpha}^\top, \bm{\xi}^\top, \bm{\varsigma}^\top)^\top$.

Let $\bm{\theta}_{\mathcal{T},0}$ be the transformed true parameters vector based on $\bm{\alpha}^0, \bm{\eta}^0, \bm{\beta}^0$ for a given tree $\mathcal{T}$. Let $\mathcal{S}_N = \{i\in \{2,3,\dots, N\}; \Vert \bm{\varsigma}_i \Vert \neq 0\}$ and $\mathcal{N}_N = \{i\in \{2,3,\dots, N\}, \Vert \bm{\varsigma}_i \Vert = 0\}$, which represent the nonzero and zero sets of $\bm{\varsigma}_i$, for $i=2,3\dots, N$ and have penalties applied in the objective function. Denote $\bm{\varsigma}_{(1)}$ and $\bm{\varsigma}_{(2)}$ as the parameter vectors for $\mathcal{S}_N$ and $\mathcal{N}_N$, respectively. Let $\mathcal{S}_{T}$ and $\mathcal{N}_T$ be the nonzero and zero sets of $\xi_t$. And denote $\bm{\xi}_{(1)}$ and $\bm{\xi}_{(2)}$ as the parameter vectors for $\mathcal{S}_T$ and $\mathcal{N}_T$, respectively. Furthermore, denote $\bm{\theta}_{\mathcal{T},0} = (\bm{\theta}_{\mathcal{T},1,0}^\top, \bm{\theta}_{\mathcal{T},2,0}^\top)^\top$, where  $\bm{\theta}_{\mathcal{T},1,0}=\left(\bm{\alpha}^\top,\bm{\varsigma}_{1}^\top,\bm{\varsigma}_{\left(1\right)}^\top,\bm{\xi}_{\left(1\right)}^\top\right)^\top$ is the true nonzero parameters vector and $\bm{\theta}_{\mathcal{T},2,0}=\bm{0}$ is the zero parameters vector. $d=2^{-1}\min\left\{ \Vert\bm{\varsigma}_{i}\Vert,i\in\mathcal{S}_{N},\vert\xi_{t}\vert,t\in\mathcal{S}_{T}\right\} $,
represents the signal.

Let $\tilde{\bm{X}}=\text{diag}(\bm{x}_{1},\bm{x}_{2,},\dots,\bm{x}_{N})$,
where $\bm{x}_{i}=\left(\bm{x}_{i1},\dots x_{iT}\right)^{\top}$, $\mathbb{X}=\tilde{\bm{X}}\left(\tilde{\bm{H}}^{-1}\otimes\bm{I}_{p}\right)$ and $\mathbb{M}=\bm{1}_{N} \otimes \bm{M}$. Let $\mathbb{U}=\left(\bm{Z},\mathbb{X},\mathbb{M}\right) = \left(\mathbb{U}_1,\mathbb{U}_2\right)$, where $\mathbb{U}_1$ is the design matrix corresponding to $\bm{\theta}_{\mathcal{T},1,0}$, and $\mathbb{U}_{2}$ is the design matrix corresponding $\bm{\theta}_{\mathcal{T},2,0}$. {\bc  In particular,  $\mathbb{U}_1$ is the submatrix of $\mathbb{U}$ corresponding to $\bm{\alpha}$, and nonzero $\bm{\varsigma}_i$'s, and nonzero $\xi_t$'s. Then, the negative loglikelihood can be written as a matrix form for a Poisson regression model,
$$
 l\left(\bm{\theta}\right)=\frac{1}{NT}\left(-\bm{y}^{\top}\left(\bm{Z}\bm{\alpha}+\mathbb{X}\bm{\varsigma}+\mathbb{M}\bm{\xi}\right)+\bm{1}^{\top}\bm{\mu}\left(\bm{\theta}\right)\right)=\frac{1}{NT}\left(-\bm{y}^{\top}\mathbb{U}\bm{\theta}+\bm{1}^{\top}\bm{\mu}\left(\bm{\theta}\right)\right),
 $$
 where $\bm{\mu}(\bm{\theta})$ is the vector of expected values of $y_{it}$ for $i=1,\dots,N$ and $t=1,\dots, T$ evaluated at $\bm{\theta}$, which is $n_{it}\exp\left(\bm{z}_{it}^{\top}\bm{\alpha} + \bm{x}_{it}^{\top}\bm{h}_{i}\vvarsigma+\bm{m}_t^{\top}\vxi\right)$ from the definition of a Poisson random variable. Then, the optimization problem in \eqref{Eq:obj_new} can be considered as a penalized regression problem. 
 When we know the sparsity of $\bm{\theta} = (\bm{\theta}_1^\top, \bm{\theta}_2^\top)^\top$, that is, if we know that $\bm{\theta}_{2} = \bm{0}$, then we can have the oracle estimator \citep{fan2001variable, fan2011nonconcave} of $\bm{\theta}_{1}$ by minimizing 
 $$
 l^{or}\left(\bm{\theta}_1\right)=\frac{1}{NT}\left(-\bm{y}^{\top}\mathbb{U}_1\bm{\theta}_1+\bm{1}^{\top}\bm{\mu}\left(\bm{\theta}_{1}\right)\right),
 $$
 which becomes a traditional Poisson regression problem. In \cite{fan2011nonconcave}, they discussed the conditions for penalized regression problems for generalized linear regression models in detail. Our model can be considered under their model structure after transformations. The existence of the oracle estimator is an important part in these problems, thus, the design matrix $\mathbb{U}_1$ should have conditions to guarantee the existence of the oracle estimator.  We followed the study in \cite{fan2011nonconcave} and adjusted their conditions to our model. 
}

Furthermore, for any positive numbers, $x_T$ and $y_T$ , $x_T \gg y_T$ means that $x_T^{-1}y_T = o(1)$. And for a vector $\bm{z}$, $\Vert \bm{z}\Vert_\infty = \sup_j \vert z_j\vert$.

Below are the assumptions.

\begin{enumerate}[label={(C\arabic*)}]
{\bc
\item The design matrix $\mathbb{U}$ satisfies that 
$\lambda_{\min}\left[\frac{1}{NT}\mathbb{U}_{1}^{\top}\mathbb{U}_{1}\right]\geq c_{1}$
 for some positive constant $c_1$, where $\lambda_{\min}(\cdot)$ is the corresponding minimum  eigenvalue. In addition, $\Vert \bm{x}_{it}\Vert_{\max}$ and $\Vert\bm{z}_{it}\Vert_{\max}$ are bounded, where $\Vert \cdot \Vert_{\max}$ denotes the largest element. 

\item $\bm{\mu}\left(\bm{\delta}\right)$
is bounded by some constants $M_{1}$ and $M_{2}$ with $M_{1}<M_{2}$ for $\bm{\delta}\in\mathcal{N}_{0}$, where $\bm{\mu}(\bm{\delta})$ is the vector of expected values of $y_{it}$ for $i=1,\dots,N$ and $t=1,\dots, T$, $\mathcal{N}_{0}=\left\{ \bm{\delta}\in\mathbb{R}^{s}:\Vert\bm{\delta}-\bm{\theta}_{\mathcal{T},1,0}\Vert_{\infty}\leq d\right\}$ and $s$ is the dimension of $\bm{\theta}_{\mathcal{T},1,0}$, and $\bm{\theta}_{\mathcal{T},1,0}$ is the subvector of nonzero parameters. 

\item $d\gg\max\left(\lambda_{1},\lambda_{2}\right)$, $\lambda_1\gg  uN_0^{-1/2}$, and $\lambda_2\gg uN_0^{-1/2}$,
where $uN_{0}^{-1/2}=o\left(1\right)$, $u\gg\sqrt{\log N_{0}}$, and $N_0 = NT$.

\item For any two locations $i$ and $i^\prime$ in the given connected network $\mathcal{G} = (\mathcal{V},\mathcal{E}_0)$ , if they are from the same cluster, then there exists a path connecting them such that all
locations on the path belong to the same cluster. 
}
\end{enumerate}

{\bc Conditions (C1) - (C2) are commonly used conditions for the design matrix in penalized problems \citep{fan2011nonconcave,wang2019spatial,kwon2012large} to guarantee the existence of the oracle estimator. It is reasonable to assume that the elements in the design matrix are bounded, which is also used in \cite{kwon2012large} when discussing logistic regression models. The assumption for the minimum eigenvalue is needed for regression problems related to the Fisher information. In Condition (C2), $\bm{\mu}(\cdot)$ represents the vector of the expected value of $y_{it}$, and it is reasonable to assume that the expected values of Poisson random variables are bounded in the small ball of the true parameter space. If we assume that the true parameter space is bounded, then Condition (C2) holds if assuming the elements in the design matrix are bounded as in Condition (C1). Note that in \cite{fan2011nonconcave}, they discussed general theoretical properties for exponential families. In their Condition 4, they have the assumptions for the design matrix, which included the minimum and maximum eigenvalues of the design matrix. In particular, the assumption of the largest eigenvalue was defined within $\mathcal{N}_0$. Here, since we assumed $\Vert \bm{x}_{it}\Vert_{\max}$ and $\Vert\bm{z}_{it}\Vert_{\max}$ are bounded, and the $\bm{\mu}(\bm{\delta})$ is bounded, Condition (C1) and Condition (C2) together can guarantee the Condition 4 in \cite{fan2011nonconcave}. 

Condition (C3) for tuning parameters $\lambda_1$, $\lambda_2$ and the minimum signal $d$ guarantee the oracle properties of the estimator, which are adjusted based on those in \cite{fan2011nonconcave}. The assumption between the signal ($d$) and tuning parameters is used in models using penalty functions, such as lasso, SCAD and MCP \citep{fan2001variable, kim2008smoothly, kwon2012large, zhang2010nearly}, in order to identify the nonzero values in parameters. In these type of problems, the theoretical properties of estimators are discussed when tuning parameters go to zero. Thus, we have $uN_0^{-1/2} = o(1)$ in Condition (C3). The assumption of $u \gg \sqrt{\log N_0}$ is due to that the number of parameters in the model is a function of $N_0$. The assumptions about $\lambda_1\gg  uN_0^{-1/2}$ and $\lambda_2\gg uN_0^{-1/2}$ are from the proof as Theorem \ref{thm1}, which is due to the number of parameters in the model divergences and has been used in \cite{fan2011nonconcave}.

Condition (C4) is used in \cite{zhang2019distributed} to ensure that the locations within the same cluster are not separated by other clusters in the graph $\mathcal{G}$. In Figure \ref{fig:mst_examples}, we illustrate the role of MST. And if the MST is like Figure \ref{fig:mst1}, we know that we can recover the true cluster structure. Since the recovery of the cluster structure depends on the structure of the MST, Condition (C4) requires that there exists a path or MST such that the true cluster structure can be identified by identifying nonzero edges.}  We can use a full graph with all pairwise connections or based on adjacency matrices, where the adjacency matrix presents the neighborhood structure of locations. In particular, if location $i$ and $i^\prime$ are neighbors, then $(i,i^\prime) \in \mathcal{E}_0$. Moreover, based on the neighborhood structure, we can also define a graph based on $k$-nearest neighbors. Condition (C4) guarantees that by removing inner-cluster connections, the original graph can be reduced to $K$ subgraphs, which correspond to $K$ clusters.

Theorem \ref{thm1} shows the theoretical properties of the estimator based on the objective function in \eqref{Eq:obj_new} for a given tree $\mathcal{T}$. The proof is provided in Appendix \ref{subsec:thm1}.

\begin{thm}
\label{thm1}
    Assume conditions (C1)-(C3) hold, $p$, $q$, $N$ and the number of change points $J$ are fixed. Then there exists a strict local minimizer $\hat{\bm{\theta}}=\left(\hat{\bm{\theta}}_{1}^{\top},\hat{\bm{\theta}}_{2}^{\top}\right)^{\top}$
of $Q\left(\bm{\theta}\right)$ in \eqref{Eq:obj_new} such that $\hat{\bm{\theta}}_{2}=\bm{0}$
with probability tending to 1 as $T\rightarrow\infty$ and $\Vert\hat{\bm{\theta}}_1-\bm{\theta}_{\mathcal{T},1,0}\Vert=O_{P}\left(\sqrt{s/N_{0}}\right)$, where $N_0 = NT$.
\end{thm}

\begin{rem}
    Theorem \ref{thm1} gives the sparsity and consistency of the estimator for a given tree $\mathcal{T}$. The sparsity property indicates that the cluster structure $\mathcal{C}_{\mathcal{T}}$ based on $\mathcal{T}$ can be recovered. But $\mathcal{C}_{\mathcal{T}}$ may not be the same as $\mathcal{C}^0$. {\bc Under $\mathcal{T}$, the corresponding tree structure can be the tree in Figure \ref{fig:mst2}. Under this $\mathcal{T}$, $\mathcal{S}_N$ represents two edges with the red cross. Theorem \ref{thm1} guarantees that we can estimate these two nonzero edges consistently, which presents three clusters. }
\end{rem}

Lemma 2 in Appendix \ref{subsec:lemma2} discussed the property of the adaptive MST, $\mathcal{T}^*$, constructed based on weights in \eqref{Eq:weight}. Lemma 2 implies that the cluster structure $\mathcal{C}_{\mathcal{T}^*}$ based on $\mathcal{T}^*$ will be the same as the true cluster structure $\mathcal{C}^0$ with probability approaching 1. {\bc Under $\mathcal{T}^*$, the corresponding tree structure can be the tree in Figure \ref{fig:mst1}. $\mathcal{S}_N^*$ is the index of nonzero $\bm{\varsigma}_i$ based on $\mathcal{T}^*$, which represents the one edge with the red cross. If the estimated index $\hat{\mathcal{S}}^*_N$ is consistent, it indicates that the red cross edge is identified and the cluster structure is identified.} Under $\mathcal{T}^*$, {\bc let $\bm{\theta}_0^* = (\bm{\theta}_{1,0}^{*\top}, \bm{\theta}_{2,0}^{*\top})^\top$ be the transformed true parameters, where $\bm{\theta}^*_{1,0}$ is the subvector of nonzero parameters.} We have the following theorem to summarize the theoretical properties of the adaptive tree based estimator based on \eqref{Eq:obj_new_final}. The proof is provided in Appendix \ref{subsec:thm2}. 
\begin{thm}
\label{thm2}
Assume conditions (C1) - (C4) hold,  $p$, $q$, $N$ and the number of change points $J$ are fixed. Based on the MST, $\mathcal{T}^{*}$, the adaptive tree based estimator $\hat{\bm{\theta}}_{1}$
satisfies that $\Vert\hat{\bm{\theta}}_{1}-\bm{\theta}_{1,0}^*\Vert=O_{P}\left(\sqrt{s/N_{0}}\right)$ and $\hat{\bm{\theta}}_2 = \bm{0}$ as $T\rightarrow \infty$. Furthermore, the sparsity property is that 
$\lim_{T\rightarrow \infty} P\left(\hat{\mathcal{S}}_N^*=\mathcal{S}_N^*\right)=1$ and $\lim_{T\rightarrow \infty} P\left(\hat{\mathcal{S}}_T=\mathcal{S}_T\right)=1$.
\end{thm}

\begin{rem}
    In Theorem \ref{thm2}, $\bm{\theta}_{1,0}^*$ depends on the structure of $\mathcal{T}^*$. 
    Due to the one-to-one mapping relationship between $\bm{\varsigma}$ and $\bm{\beta}$, when the sparsity structure of $\bm{\varsigma}$ is recovered, the cluster structure of $\bm{\beta}$, $\mathcal{C}^0$, is also recovered. 
\end{rem}

{\bc

Recall that, $\mathcal{S}_N^*$ is index set of nonzero $\bm{\varsigma}_i$, which represent the nonzero differences between $\bm{\beta}_i$ and $\bm{\beta}_{i^{\prime}}$. The number of clusters $K$ is the number of nonzero index in $\mathcal{S}_N^*$ plus 1. Thus, if $\mathcal{S}_N^*$ is recovered, the number of clusters is properly estimated. Similarly, $\mathcal{S}_T$ is the index set of nonzero $\xi_t$, which represents the difference between $\eta_t - \eta_{t-1}$. Thus, if $\mathcal{S}_T$ is recovered, the number of change points is properly estimated. We have the following corollary to present the results. 

\begin{cor}
    Under the assumption in Theorem \ref{thm2},  we have $\lim_{T\rightarrow \infty} P(\hat{K} = K) = 1$ and $\lim_{T\rightarrow \infty} P(\hat{J} = J)=1$, where $\hat{K}$ is the estimated number of clusters, and $\hat{J}$ is the estimated number of change points.
\end{cor}

Note that the order in our theorems is the theoretical error upper bound that can be achieved with proper $\lambda_1$ and $\lambda_2$ selection. In practice, the numerical error bound would be affected by the selected values of $\lambda_1$ and $\lambda_2$.
}

\section{Simulation}
\label{Sec:sim}

In this section, we use several simulation examples to evaluate the performance of our proposed estimator. First, we compare our proposed method to the traditional approaches in Section \ref{subsec:sim_comp_scan} when the number of clusters is 2. Then, we evaluate our proposed method for different setups in Section \ref{subsec:sim_ncluster}.

\subsection{Comparison to existing approaches}
\label{subsec:sim_comp_scan}

In this section, we will compare our proposed method to spatial scan approaches {\bc and a change point detection approach when the number of clusters is 2 and the number of change points is 1}. To evaluate the performance of recovering cluster structure, we use the estimated number of clusters ($\hat{K}$) and adjusted Rand Index (ARI)\citep{hubert1985comparing}. ARI computes the degree of overlapping between two partitions and is widely used to evaluate the clustering performance against the ground truth. ARI ranges from -1 to 1, and its value is equal to 1 only if a partition is completely identical to the intrinsic structure; its value is equal to -1 when the partition is completely different from the intrinsic structure; its value is close to 0 for a random partition. {\bc To evaluate the performance of detecting change points, Hausdorff distance \citep{delfour2011shapes} is used to calculate the distance between the true change point set and the estimated change point set. A zero value means that these two sets are matched. Since there is only one changed point under this simple case, the Hausdorff distance measures the absolute value between the true time of the change point and the estimated time.}

We simulate our data from the model in \eqref{eq:rate_intercept} with clustered intercepts. The population value $n_{it}$ is generated from a lognormal distribution with parameters of center 10 and standard deviation 0.7. $z_{it}$'s, for $i=1,2,\dots, N$ and $t=1,\dots, T$, are simulated from the standard normal distribution, and the regression coefficient  $\alpha = 0.5$.  For time effect $\eta_t$, we consider one change point with $\eta_t = 0$ for $t=1,\dots, 10$ and $\eta_t = -0.5$ for $t=11,\dots, T=20$. For location-specific effect $\beta_i$, we consider two clusters with $N=100$ and the two spatial structures as shown in Figure \ref{fig:ncluster2_structures}.  Figure \ref{fig:ncluster2_v1} is a $10\times 10$ grid lattice and Figure \ref{fig:ncluster2_v2} contains $100$ random locations uniformly generated from (-1,1). In both structures, the center cluster is smaller than the outside cluster. {\bc We consider two sets of values of $\beta_i$ as below:
\begin{itemize}[leftmargin=2.5cm]
\item[Setting 1:] $\beta_i = -7$ for $i$ in the smaller cluster and $\beta_i = -7.5$ for $i$ in the big one;
\item[Setting 2:] $\beta_i = -7$ for $i$ in the smaller cluster and $\beta_i = -7.25$ for $i$ in the big one. 
\end{itemize}
}

\begin{figure}[H]
\centering
\begin{subfigure}{.44\textwidth}
  \centering
  \includegraphics[width=0.95\linewidth]{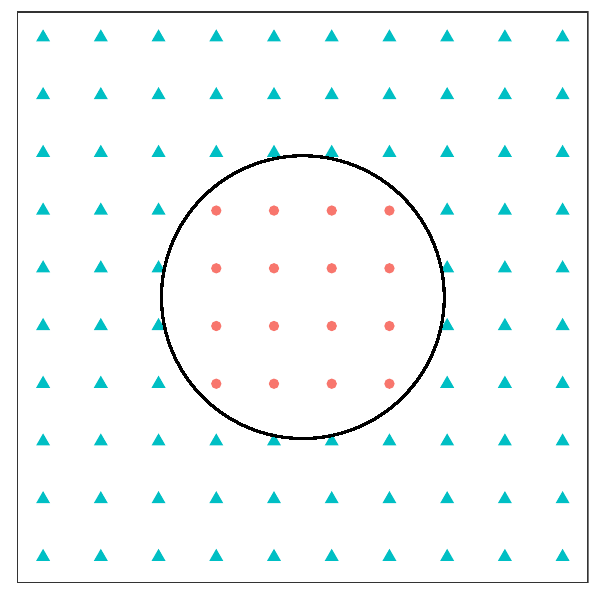}
  \caption{Cluster structure for a $10\times 10$ grid lattice}
  \label{fig:ncluster2_v1}
\end{subfigure}%
\begin{subfigure}{.44\textwidth}
  \centering
  \includegraphics[width=0.95\linewidth]{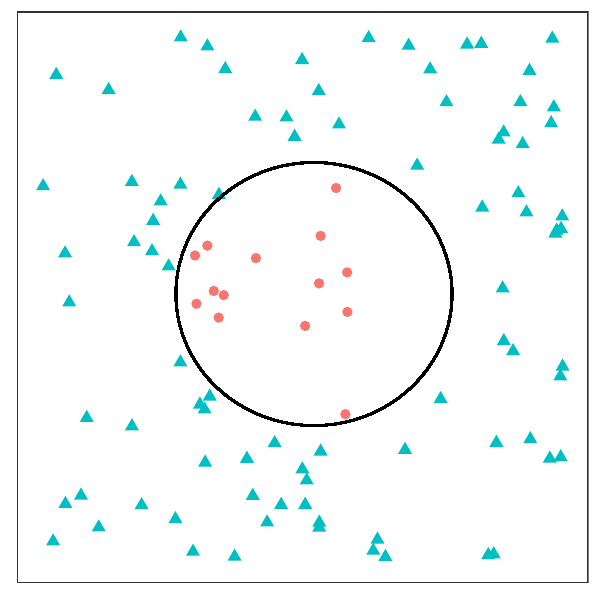}
  \caption{Cluster structure for 100 random locations.}
  \label{fig:ncluster2_v2}
\end{subfigure}
\caption{True cluster structures for simulated data. The color and shape show the underlying cluster structure. The solid lines separate the 2 clusters.  }
\label{fig:ncluster2_structures}
\end{figure}

We compare our proposed adaptive method (denoted as ``adaptive") with several spatial scan statistics for functional data implemented in the R package  {\it HDSpatialScan} \citep{frevent2022r}, including nonparametric functional scan statistic (NPFSS) \citep{smida2022wilcoxon}, parametric functional scan statistic (PFSS) \citep{frevent2021detecting}, distribution-free functional scan statistic (DFFSS) \citep{frevent2021detecting} and univariate rank-based functional scan statistic (URBFSS) \citep{frevent2023investigating}. Besides these approaches, we consider two methods in the R package {\it scanstatistics} \citep{allevius2018scanstatistics}, one is the space-time permutation scan statistic (PSS) \citep{kulldorff2005space} and the other one is population-based Poisson scan statistic (PPSS)\citep{kulldorff2001prospective}. {\bc The change point detection method in \cite{jonas2020} is denoted as ``rmcp". Note that ``rmcp" cannot handle change point detection for a multivariate time series and needs to specify the number of change points. In our simulation, we set the number of change points as the true value of 1, and use the method to detect change points for each location separately.} We also consider the method using MST based on locations, denoted as ``tree".

{\bc Tables \ref{tab:res_nc2_s1} and \ref{tab:res_nc2_s2} show the average ARI, average estimated number of cluster $\hat{K}$ and average Hausdorff distance based on 100 simulations for different methods, along with standard deviation values in parentheses. Note that the scan approaches don't have a change point detection component, thus the Hausdorff distance is not reported. And ``rmcp" is only for change point detection without spatial clustering component, thus ARI and $\hat{K}$ are not reported. We observe that approaches based on scan statistics cannot identify the cluster structure well. This is because these approaches do not consider the effects of covariates and the changes in the temporal domain. When ignoring the spatial pattern and the covariates, ``rmcp" cannot detect the change points correctly. 
When comparing the adaptive method to the tree based approach, we can achieve better performance in terms of estimating cluster structure. When comparing the results for these two settings of $\beta_i$, the ARI of the adaptive approach in Setting 2 is smaller than that in Setting 1, which is because the cluster difference is smaller in Setting 2.}


\begin{table}[H]
\centering
\caption{Average ARI, average $\hat{K}$ and average Hausdorff distance over 100 simulations under spatial grid lattice in Figure \ref{fig:ncluster2_v1}. Standard deviation values are in parentheses.}
\label{tab:res_nc2_s1}
\begin{tabular}{c|ccc|ccc}
  \hline
  & \multicolumn{3}{c|}{Setting 1} &\multicolumn{3}{c}{Setting 2}\\
  \hline
method & ARI & $\hat{K}$ & dist & ARI & $\hat{K}$ & dist \\ 
  \hline
tree & 0.569(0.052) & 8.75(1.64) & 0.24(1.24) & 0.388(0.241) & 6.07(3.21) & 0.24(1.37) \\ 
  adaptive & 0.994(0.064) & 2.01(0.10) & 0.11(0.78) & 0.767(0.273) & 2.60(0.71) & 0.08(0.80) \\ 
  NPFSS & 0.517(0.261) & 1.84(0.44) & - & -0.005(0.017) & 1.08(0.27) & - \\ 
  PFSS & 0.019(0.109) & 1.03(0.17) & - & -0.001(0.006) & 1.02(0.14) & - \\ 
  DFFSS & 0.003(0.028) & 1.02(0.14) & - & 0.001(0.008) & 1.01(0.10) & - \\ 
  URBFSS & 0.192(0.281) & 1.37(0.49) & - & 0.042(0.156) & 1.08(0.27) & - \\ 
  PSS & 0.669(0.023) & 2.00(0.00) & - & 0.525(0.206) & 2.00(0.00) & - \\ 
  PPSS & 0.669(0.023) & 2.00(0.00) & - & 0.525(0.206) & 2.00(0.00) & - \\ 
  rmcp & - & - & 3.29(0.20) & - & - & 3.16(0.22) \\ 
   \hline
\end{tabular}
\end{table}


\begin{table}[H]
\centering
\caption{Average ARI, average $\hat{K}$ and average Hausdorff distance over 100 simulations under spatial random locations in Figure \ref{fig:ncluster2_v2}. Standard deviation values are in parentheses.}
\label{tab:res_nc2_s2}
\begin{tabular}{c|ccc|ccc}
  \hline
  & \multicolumn{3}{c|}{Setting 1} &\multicolumn{3}{c}{Setting 2}\\
  \hline
method & ARI & $\hat{K}$ & dist & ARI & $\hat{K}$ & dist \\ 
  \hline
tree & 0.179(0.011) & 9.68(0.74) & 0.05(0.50) & 0.188(0.016) & 6.25(1.38) & 0.56(2.07) \\ 
  adaptive & 1.000(0.000) & 2.00(0.00) & 0.05(0.50) & 0.752(0.250) & 2.52(0.59) & 0.28(1.41) \\ 
  NPFSS & 0.800(0.123) & 1.98(0.14) & - & 0.027(0.140) & 1.06(0.24) & - \\ 
  PFSS & 0.143(0.307) & 1.20(0.40) & - & 0.002(0.019) & 1.02(0.14) & - \\ 
  DFFSS & 0.003(0.030) & 1.01(0.10) & - & 0.000(0.000) & 1.00(0.00) & - \\ 
  URBFSS & 0.079(0.175) & 1.26(0.44) & - & 0.008(0.038) & 1.04(0.20) & - \\ 
  PSS & 0.516(0.055) & 2.00(0.00) & - & 0.400(0.172) & 2.00(0.00) & - \\ 
  PPSS & 0.516(0.055) & 2.00(0.00) & - & 0.400(0.172) & 2.00(0.00) & - \\ 
  rmcp & - & - & 3.23(0.20) & - & - & 3.13(0.22) \\ 
   \hline
\end{tabular}
\end{table}

\subsection{Evaluation of model performance}
\label{subsec:sim_ncluster}

In this section, we will evaluate the performance of our proposed adaptive approach when the number of clusters is 5, and the number of change points is 2. We will compare the adaptive approach to the tree based approach.  Besides the estimated number of clusters $\hat{K}$ and ARI, we also report the estimated number of changed points ($\hat{J}$) and the F1 score \citep{sasaki2007truth} to evaluate the performance of detecting changed points. F1 score measures classification accuracy (changed or nonchanged) with the highest value of 1 and the smallest value of 0. The higher the value of the F1 score is, the better the accuracy of temporal trend detection is.  The F1 score is defined as $F1 = \frac{2PR}{P+R}$, where $P$ is precision and $R$ is recall,  {\bc $P = \frac{\text{number of true detected change-points}}{\text{number of total detected change-points}}$ and  $R = \frac{\text{number of true detected change-points}}{\text{number of total change-points}}$}.
To evaluate the estimation accuracy, we report the root mean square error (RMSE) for estimating $\bm{\beta}$, $\bm{\eta}$ and $\bm{\alpha}$. RMSE is defined as $\sqrt{\frac{1}{\vert \bm{\theta}\vert}\Vert \hat{\bm{\theta}} -\bm{\theta}\Vert^2}$, where $\bm{\theta}$ can be $\bm{\beta}$, $\bm{\eta}$ or $\bm{\alpha}$,  $\hat{\bm{\theta}}$ is an estimate of $\bm{\theta}$ and $\vert \bm{\theta}\vert$ is the dimension of $\bm{\theta}$.

Similar to the setup in Section \ref{subsec:sim_comp_scan}, we consider two spatial cluster structures: one is a grid lattice, and the other one is based on random locations.  Figure \ref{fig:ncluster5_structures} shows the two underlying true cluster structures.
\begin{figure}[H]
\centering
\begin{subfigure}{.45\textwidth}
  \centering
  \includegraphics[width=0.95\linewidth]{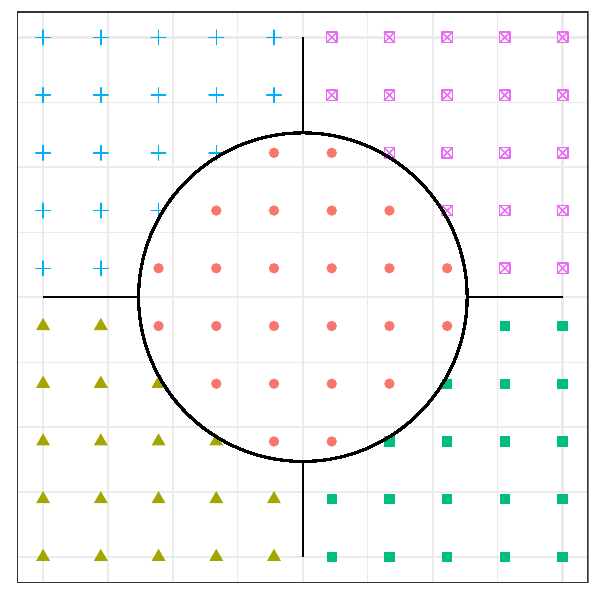}
  \caption{Cluster structure for a $10\times 10$ grid lattice}
  \label{fig:ncluster5_v1}
\end{subfigure}%
\begin{subfigure}{.45\textwidth}
  \centering
  \includegraphics[width=0.95\linewidth]{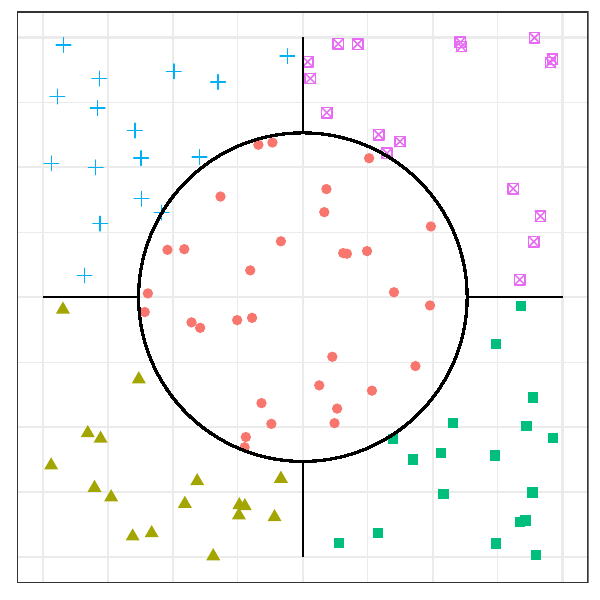}
  \caption{Cluster structure for 100 random locations.}
  \label{fig:ncluster5_v2}
\end{subfigure}
\caption{True cluster structures for simulated data. The color and shape show the underlying cluster structure. The solid lines separate the 5 clusters. }
\label{fig:ncluster5_structures}
\end{figure}

 We simulate data as follows: let $\bm{x}_{it}=[1, x_{it}]^{\top}$ with $x_{it}$'s are generated from $N(0,1)$ and $z_{it}$'s are generated from $N(0,1)$, for $i=1,2,\dots, 100$ and $t=1,2,\dots, 25$. For model parameters, we set $\alpha = 0.5$, and $\bm{\beta}_i$ takes two set of parameters for five clusters $\mathcal{C}_1$, $\mathcal{C}_2$, $\mathcal{C}_3$, $\mathcal{C}_4$ and $\mathcal{C}_5$: 
 
Setting 1: $\bm{\beta}_{i}=(-8,-1)^{T}$ if $i\in\mathcal{C}_{1}$; $\bm{\beta}_{i}=(-7.7,-0.5)^{T}$ if $i\in\mathcal{C}_{2}$; $\bm{\beta}_{i}=(-7.5,0)^{T}$ if $i\in\mathcal{C}_{3}$; $\bm{\beta}_{i}=(-7.2,0.5)^{T}$ if $i\in\mathcal{C}_{4}$; $\bm{\beta}_{i}=(-7,1)^{T}$ if $i\in\mathcal{C}_{5}$;

Setting 2: $\bm{\beta}_{i}=(-8,-0.5)^{T}$ if $i\in\mathcal{C}_{1}$; $\bm{\beta}_{i}=(-7.7,-0.25)^{T}$ if $i\in\mathcal{C}_{2}$; $\bm{\beta}_{i}=(-7.5,0)^{T}$ if $i\in\mathcal{C}_{3}$; $\bm{\beta}_{i}=(-7.2,0.25)^{T}$ if $i\in\mathcal{C}_{4}$; $\bm{\beta}_{i}=(-7,0.5)^{T}$ if $i\in\mathcal{C}_{5}$.  

The cluster difference of setting 2 is smaller than that of setting 1. We assume there are 2 changed points, one is at $t=5$ with a change from 0 to -0.5, and the other one is at $t=15$ with a change from -0.5 to -0.8.

Table \ref{tab:res_n5_set1} and Table \ref{tab:res_n5_set2} show the average ARI, $\hat{K}$, F1 score and $\hat{J}$ over 100 simulations, along with standard deviations in parentheses, for two different spatial cluster structures. Figure \ref{fig:rmse_s1} and Figure \ref{fig:rmse_s2} show the RMSE for estimating $\alpha$, $\bm{\beta}$ and $\bm{\eta}$ under two spatial cluster structures. It can be seen that the structure of trees does not affect change point detection but significantly affects cluster identification. The adaptive method has larger ARI values and smaller RMSE values compared to the tree based approach. The tree based approach tends to identify more clusters. This is because that the initial MST may not reflect the true spatial structure. 

\begin{table}[H]
\centering
\caption{Summary of results over 100 simulations under spatial grid lattice in Figure \ref{fig:ncluster5_v1}}
\label{tab:res_n5_set1}
\begin{tabular}{c|ccccccc}
  \hline
 & algorithm & ARI & $\hat{K}$ & F1 & dist &$\hat{J}$ \\ 
  \hline
  \multirow{2}{*}{S1}& tree & 0.597(0.025) & 16.700(1.314) & 0.998(0.020) & 0.040(0.400) & 2.010(0.100) \\ 
   & adaptive & 0.998(0.009) & 5.010(0.100) & 0.998(0.020) & 0.050(0.500) & 2.010(0.100) \\ 
   \hline
   \multirow{2}{*}{S2}&
    tree & 0.576(0.022) & 14.620(0.850) & 0.996(0.028) & 0.080(0.563) & 2.020(0.141) \\ 
   & adaptive & 0.962(0.049) & 5.070(0.293) & 1.000(0.000) & 0.000(0.000) & 2.000(0.000) \\ 
   \hline
\end{tabular}
\end{table}

\begin{table}[H]
\centering
\caption{Summary of results over 100 simulations under spatial random locations in Figure \ref{fig:ncluster5_v2}}
\label{tab:res_n5_set2}
\begin{tabular}{c|ccccccc}
  \hline
 & algorithm & ARI & $\hat{K}$ & F1 & dist & $\hat{J}$ \\ 
  \hline
  \multirow{2}{*}{S1}& tree & 0.763(0.004) & 13.070(1.018) & 1.000(0.000) & 0.000(0.000) & 2.000(0.000) \\ 
   & adaptive & 0.999(0.004) & 5.010(0.100) & 1.000(0.000) & 0.000(0.000) & 2.000(0.000) \\
   \hline
   \multirow{2}{*}{S2}&
   tree &  0.763(0.009) & 12.030(0.223) & 1.000(0.000) & 0.000(0.000) & 2.000(0.000) \\ 
   & adaptive &  0.985(0.016) & 5.000(0.000) & 1.000(0.000) & 0.000(0.000) & 2.000(0.000) \\ 
   \hline
\end{tabular}
\end{table}

\begin{figure}[H]
\centering
\begin{subfigure}{.47\textwidth}
  \centering
  \includegraphics[width=0.95\linewidth]{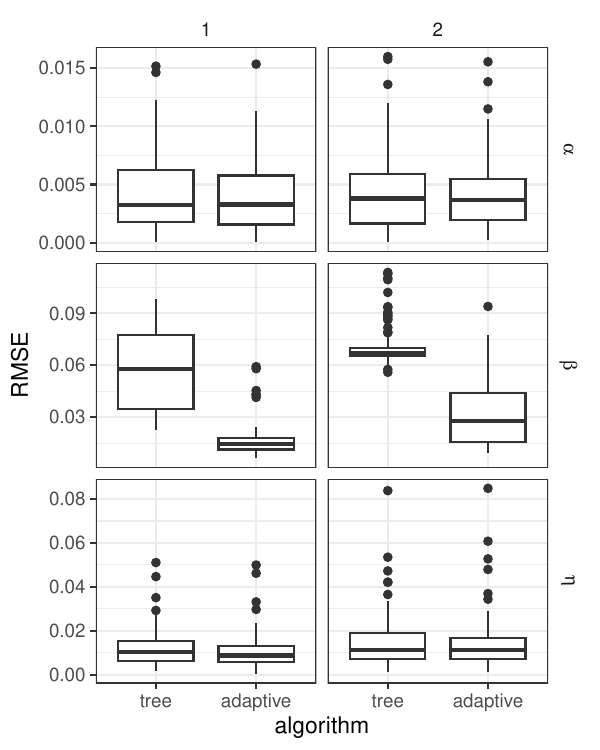}
  \caption{RMSE for spatial lattice grid.}
  \label{fig:rmse_s1}
\end{subfigure}%
\begin{subfigure}{.47\textwidth}
  \centering
  \includegraphics[width=0.95\linewidth]{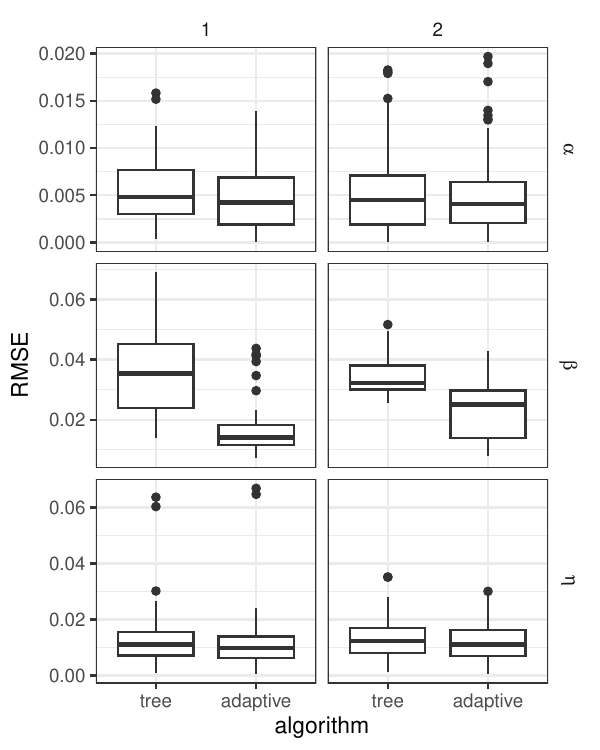}
  \caption{RMSE for spatial random locations}
  \label{fig:rmse_s2}
\end{subfigure}
\caption{RMSE of different parameters for spatial lattice grid (left) spatial random locations (right). In (a) and (b), the left three figures are for $\alpha$, $\bm{\beta}$ and $\bm{\eta}$ under Setting 1 and the right three figures are for $\alpha$, $\bm{\beta}$ and $\bm{\eta}$ under Setting 2. }
\label{fig:rmse}
\end{figure}

{\bc We also report the computational time for Setting 2 under the spatial random locations. The computational time is based on 30 $\lambda_1$ values and 30 $\lambda_2$ values and implemented on a cluster server with Intel Xeon E5-2650 v4 (2.2GHz). Note that the running time also depends on the machines. The real data analysis was implemented on a MacBook Pro with Apple M2 Chip, which is faster. Figure \ref{fig:comp_time} gives the computational time in minutes for the tree based approach and the adaptive approach. It can be seen that the computational time of adaptive approach is longer due to an extra estimation step. As discussed in Section \ref{subsec:algorithm}, we use a two-step procedure to select tuning parameters, and the computational time is proportional to the number of tuning parameters. The computational time would be proportional to the product of the number of two tuning parameters when using a bivariate grid, which would lead to a long computational time.} 

\begin{figure}
    \centering
    \includegraphics[width=0.5\linewidth]{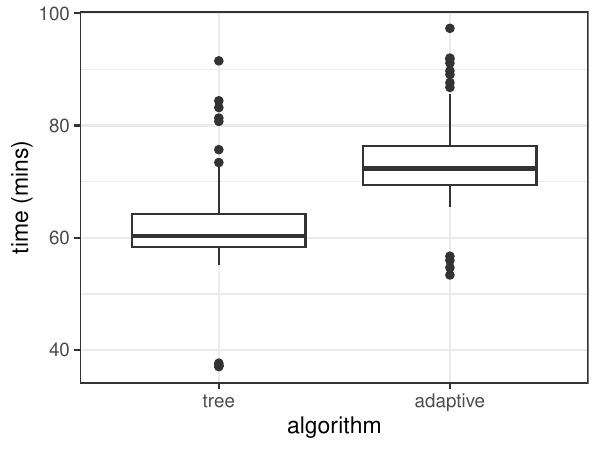}
    \caption{Computational time in minutes for Setting 2 under the spatial random locations }
    \label{fig:comp_time}
\end{figure}

\section{Real data analysis}
\label{Sec:example}

In this section, we apply our proposed method to data from Surveillance, Epidemiology, and End Results (SEER) program \citep{seerdata, duggan2016surveillance}. SEER program provides cancer statistics for the U.S. population, with data collected in several states since 1973. Currently, there are 22 registries, including states and metropolitan areas, which can cover approximately 47.9\% of the U.S. population. SEER data are widely used in different studies for different cancers, such as breast cancer \citep{guha2022incidence}, lung cancer \citep{tolwin2020gender}, colon and rectum cancer \citep{daly2019surveillance}, and respiratory system cancer \citep{azevedo2021mspock}.  There are also some existing studies about cluster detection in SEER data. For example, \cite{geng2022bayesian} conducted an analysis of personal-level data on respiratory cancer to find clusters in Louisiana counties. \cite{chernyavskiy2019heterogeneity} used multilevel age-period cohort models to analyze the heterogeneity of colon and rectum cancer incidence across different counties and age groups. Exploring the spatial cluster structure of incidence can help to allocate resources for interventions and enhance the understanding of spatially varying risk factors \citep{chernyavskiy2019heterogeneity, amin2019spatial}.

As an illustration, we consider colon and rectum cancer in Iowa from 1995 to 2020. The response variable $y_{it}$ is the number of deaths with an age greater than 30 in each county. Under our proposed model, the covariate is poverty from the Small Area Income and Poverty Estimates (SAIPE) Program. The poverty data are obtained from the R package \emph{censusapi}, which has a complete time series record from 1995 to 2020. We consider the following model for the 99 counties in Iowa from the year 1995 to 2020, 
\begin{equation}
    y_{it} \sim \text{Poisson}(n_{it}\mu_{it}),\quad\log \mu_{it} = z_{it} \alpha + \beta_i + \eta_t,
\end{equation}
where $n_{it}$ is the population for age greater than 30 for county $i$ at year $t$, $z_{it}$ presents the estimated poverty for county $i$ at year $t$ from SAIPE, $\alpha$ is the common regression coefficient, $\beta_i$ the county effect, and $\eta_t$ is the time effect. Our goal is to find the cluster structure of $\beta_i$, and to detect the change points of $\eta_t$.  {\bc In this analysis, we use 45 values of $\lambda_1$ and 45 values of $\lambda_2$ with 10 different initial trees and select the model with the smallest BIC. The computational time of these 10 models ranges from 23.88 to 26.74 minutes, which are implemented on a MacBook Pro with Chip Apple M2.}

The estimated coefficient for {\it poverty} is $\hat{\alpha} = 0.0241$ with standard error 0.0088, which indicates that as poverty increases, the expected death rate increases.  Figure \ref{fig:countyeffect} gives the estimated cluster structure of county effects. There is one isolated county, Johnson County, which has the smallest estimated $\hat{\beta}_i = -8.218$ with a standard error of 0.0537. This county is where Iowa City and the University of Iowa are located. The cluster with $\hat{\beta}_i = -8.020$ with a standard error of 0.0242 has the second smallest county effect. This cluster includes four counties: Warren, Dallas, Polk, and Story. Among these four counties, Des Moines is the capital of Iowa and is in Polk County; Ames is in Story County and it is the city where Iowa State University is located. {\bc The potential reasons for smaller intercepts, therefore, lower expected death rates in these five counties after adjusting the effect of the covariate, can be education and income. Note that the University of Iowa and Iowa State University are the two largest public universities in Iowa. The University of Iowa (Johnson County) has hospitals and a public health college; Iowa State University (Story County) also has research related to public health. Warren, Dallas, and Polk counties are the top three counties in terms of average median income during the analyzed period from the SAIPE.} Besides these two small clusters, there are two larger clusters that have larger county effects with estimated $\hat{\beta}_i = -7.822$ (with standard error 0.0180) and $\hat{\beta}_i = -7.598$ (with standard error 0.0146). {\bc The heterogeneity of county effects suggests more resources can be allocated to counties with higher death rates, as access to medical resources in cities and educational facilities can positively influence these rates.}

\begin{figure}[h]
    \centering
    \includegraphics{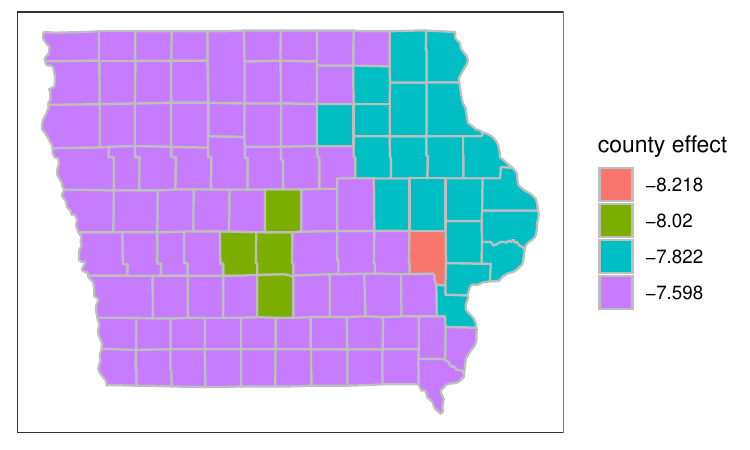}
    \caption{Estimated cluster structure of county effect}
    \label{fig:countyeffect}
\end{figure}

Two change points are detected. One was in 2004, changed from 0 to -0.2133 with a standard error of 0.02; the other was in 2012, changed from -0.2133 to -0.3765 with a standard error of 0.0202. The estimated negative changes indicate that the death rate has a large decline in 2004 and 2012. Retrospective research on medical policies can be conducted to understand the reasons behind the declines. {\bc Possible reasons for change points can be nutrition and insurance. For example, the Nutrition Labeling and Education Act in 1990 required nutrition labeling of most foods. For insurance policies, starting 2001, Medicare coverage for colonoscopy was expanded from the high-risk group to the average-risk group \citep{harewood2004colonoscopy}, and the Affordable Care Act (ACA) enacted in 2010 required both private insurers and Medicare to cover preventive services \citep{richman2016colorectal}. All of these policies are likely to contribute to the reduction in mortality from colon and rectum cancer in Iowa over the years.}

\section{Conclusion and discussions}

In this article, we proposed a new method for clustering locations based on covariate effects and identifying temporal change points for count data based on a Poisson regression model. To achieve the goal, we designed a doubly penalized likelihood approach. In the proposed approach, an adaptive minimum spanning tree (MST) was used to identify spatial clusters, and a fused penalty was used to detect change points. In the simulation study, we used several numerical examples to investigate the performance of the proposed approach. The numerical results showed that the proposed approach outperforms traditional scan statistics in recovering spatial structures. We also found that the adaptive MST can perform better than a random tree in terms of recovering spatial cluster structure. 

The idea of the proposed approach can be extended to semiparametric or nonparametric regression models, which can model the effects of covariates more flexibly, as discussed in linear regression models \citep{liu2019subgroup}. In addition, there is a potential to explore the patient level in the SEER data to identify the heterogeneity for both individual patient and spatial levels.


\begin{appendices}
\section*{Appendix}
\section{Lemmas}
\label{sec:lemmas}
In this part, we introduce two lemmas, which will be used later in the proof of Theorem 1 and Theorem 2.

\subsection{Lemma 1}
\label{subsec:lemma1}

Some notations are defined as below before introducing Lemma \ref{lem1}.

Denote $\hat{\bm{\theta}}$ as an estimator of $\bm{\theta}$. Let $\hat{\mathcal{S}}_N = \{i \in \{2,3,\dots,N\};\Vert \hat{\bm{\varsigma}}_i\Vert \neq 0 \}$, and $\hat{\mathcal{N}}_N = \{i \in \{2,3,\dots,N\};\Vert \hat{\bm{\varsigma}}_i\Vert = 0 \}$ be the nonzero and zero sets of estimates $\hat{\bm{\varsigma}}$. Denote $\hat{\bm{\varsigma}}_{(1)}$ and $\hat{\bm{\varsigma}}_{(2)}$ as the parameter vectors for $\hat{\mathcal{S}}_N$ and $\hat{\mathcal{N}}_N$, respectively.
Similarly, $\hat{\mathcal{S}}_T = \{i \in \{2,3,\dots,T\};\vert \hat{\xi}_t\vert \neq 0 \}$, and $\hat{\mathcal{N}}_T = \{i \in \{2,3,\dots,T\};\vert \hat{\xi}_i\vert = 0 \}$ be the nonzero and zero sets of estimates $\hat{\bm{\xi}}$. And denote $\hat{\bm{\xi}}_{(1)}$ and $\hat{\bm{\xi}}_{(2)}$ as the parameter vectors for $\hat{\mathcal{S}}_T$ and $\hat{\mathcal{N}}_T$, respectively.

Recall that the first column of $\tilde{\bm{H}}^{-1}$ is $\frac{1}{\sqrt{N}}\bm{1}$,
then the column in $\mathbb{X}$ corresponds to $\bm{\varsigma}_{1}$
is $\bm{X}^{*}=\left(\bm{x}_{1}^{\top},\bm{x}_{2}^{\top},\dots,\bm{x}_{N}^{\top}\right)^{\top}$.
Let $\mathbb{X}_{1}$ and $\mathbb{X}_{2}$ denote the submatrices
of $\mathbb{X}$ formed by columns $\hat{\mathcal{S}}_N$ and $\hat{\mathcal{N}}_N$. $\mathbb{M}_{1}$ and $\mathbb{\mathbb{M}}_{2}$ denote the submatrices of $\mathbb{M}$ formed by columns $\hat{\mathcal{S}}_T$ and $\hat{\mathcal{N}}_T$.  

For MCP penalty, we have 
\[
\mathcal{P}^{\prime}_{\gamma}\left(t,\lambda\right)=\begin{cases}
\lambda\text{sign}\left(t\right)-\frac{t}{\gamma}=\left(\lambda-\frac{\vert t\vert}{\gamma}\right)\text{sign}\left(t\right) & \vert t\vert\leq\gamma\lambda,\\
0 & \vert t\vert>\gamma\lambda,
\end{cases}
\]
and $\mathcal{P}_{\gamma}^{\prime}\left(0+,\lambda\right)=\lambda$. For $\bm{v}\in \mathbb{R}^s$ with $\Vert \bm{v}\Vert_0 = s$, we also introduce $\kappa_{\gamma}(\bm{v})$ as follows \citep{fan2011nonconcave},
\[
\kappa_{\gamma}\left(\bm{v}\right)=\lim_{\epsilon\rightarrow0+}\max_{j}\sup_{0<t_{1}<t_{2}\in(\vert v_{j}\vert-\epsilon,\vert v_{j}\vert+\epsilon)}-\frac{\mathcal{P}_{\gamma}^{\prime}\left(t_{2},\lambda\right)-\mathcal{P}^\prime_{\gamma}\left(t_{1}, \lambda\right)}{t_{2}-t_{1}},
\]
which corresponds to the second order derivative of $-\mathcal{P}_{\gamma}(t,\lambda)$ when the second order derivative exists. Note that, for MCP, the second-order derivative does not exist for a finite number of values of $t$, they are $0$ and $\gamma\lambda$. $\kappa_\gamma(\bm{v})=\max_{1\leq j\leq s}-\mathcal{P}^{\prime \prime}(\vert v_j\vert,\lambda)$ provided that the second-order derivative is continuous. Note that, $\mathcal{P}^{\prime \prime}(\vert v_j\vert,\lambda) = 0$ if $\vert v_j\vert >\gamma\lambda$, and $\mathcal{P}^{\prime \prime}(\vert v_j\vert,\lambda) = -\frac{1}{\gamma}$ if $0<\vert v_j\vert <\gamma\lambda.$ If all $\vert v_j\vert >\gamma\lambda$, then $\kappa_\gamma(\bm{v}) =0$. And $\kappa_\gamma(\bm{v}) =\frac{1}{\gamma}$ if some  $0< \vert v_j\vert \leq \frac{1}{\gamma}$. 

Furthermore, for a vector $\bm{u}$, we have the derivative with respect to $\bm{u}$ as
\[
\frac{d\mathcal{P}_{\gamma}\left(\bm{u},\lambda\right)}{\bm{u}} = \frac{\bm{u}}{\Vert \bm{u}\Vert}\mathcal{P}^{\prime}_{\gamma}\left(\Vert\bm{u}\Vert,\lambda\right).
\]
When considering the second-order derivative, we have 
\begin{align*}
\left[\frac{\bm{u}}{\Vert\bm{u}\Vert}\mathcal{P}_{\gamma}^{\prime}\left(\Vert\bm{u}\Vert,\lambda\right)\right]^{\prime} & =\frac{\bm{u}\bm{u}^\top}{\Vert\bm{u}\Vert^{2}}\mathcal{P}_{\gamma}^{\prime\prime}\left(\Vert\bm{u}\Vert,\lambda_{}\right)+\frac{\mathcal{P}_{\gamma_{}}^{\prime}\left(\Vert\bm{u}\Vert,\lambda_{}\right)}{\Vert\bm{u}\Vert}\bm{I}_{p}-\frac{\bm{u}\bm{u}^\top\mathcal{P}_{\gamma_{}}^{\prime}\left(\Vert\bm{u}\Vert,\lambda_{}\right)}{\Vert\bm{u}\Vert^{3}}.
\end{align*}
Note that, when $\Vert \bm{u}\Vert >\gamma\lambda$, the value is 0. Then, for $\bm{u} = (\bm{u}_1^\top,\dots, \bm{u}_s^\top)^\top$, we can define  $\kappa^\prime(\bm{u}) = \max_{1\leq j\leq s} -\left[\frac{\bm{u}_j}{\Vert\bm{u}_j\Vert}\mathcal{P}_{\gamma}^{\prime}\left(\Vert\bm{u}_j\Vert,\lambda\right)\right]^{\prime}$. And we know that if $\Vert \bm{u}_j\Vert >\gamma\lambda$ for all $j$, then $\kappa^\prime(\bm{u}) =0$.

Let $\bm{b}=\bm{Z}\bm{\alpha}+\mathbb{X}\bm{\varsigma}+\mathbb{M}\bm{\xi}$,
$\bm{\mu}\left(\bm{\theta}\right)=\left(n_{it}\exp\left(b_{it}\right)\right)_{1\le i\leq N,1\leq t\leq T}$, and $\bm{\Sigma}\left(\bm{\theta}\right)=$\\$\text{diag}\left(n_{it}\exp\left(b_{it}\right);1\leq i\leq N,1\leq t\leq T\right) $. $\bm{\mu}(\bm{\theta})$ corresponds to the expected values vector of $y_{it}$ for $i=1,\dots, N$ and $t=1,\dots, T$, and $\bm{\Sigma}\left(\bm{\theta}\right)$ is the covariance matrix.  We have the following lemma that gives sufficient conditions on the strict local minimizer of $Q(\bm{\theta})$. The proof follows that in \cite{fan2011nonconcave}.

\begin{lem}
\label{lem1}
$\hat{\bm{\theta}}$ is a strict local minimizer of $Q(\bm{\theta})$ if 

\begin{align}
\frac{1}{NT}\bm{Z}^{\top}\left(\bm{y}-\bm{\mu}(\hat{\bm{\theta}})\right) & =\bm{0} \label{eq_kkt1} \\
\frac{1}{NT}\bm{X}^{*\top}\left(\bm{y}-\bm{\mu}(\hat{\bm{\theta}})\right) & = \bm{0}\label{eq_kkt2}\\
\frac{1}{NT}\mathbb{X}_{1}^{\top}\left(\bm{y}-\bm{\mu}(\hat{\bm{\theta}})\right)-\left(\frac{\hat{\bm{\varsigma}}_{i}}{\Vert \hat{\bm{\varsigma}}_{i}\Vert}\mathcal{\mathcal{P}}_{\gamma_{2}}^{\prime}\left(\Vert\hat{\bm{\varsigma}}_{i}\Vert,\lambda_{2}\right);i\in\hat{\mathcal{S}}_{N}\right) & =\bm{0} \label{eq_kkt3}\\
\frac{1}{NT}\mathbb{M}_{1}^{\top}\left(\bm{y}-\bm{\mu}(\hat{\bm{\theta}})\right)-\left(\mathcal{\mathcal{P}}_{\gamma_{1}}^{\prime}\left(\hat{\xi}_{t},\lambda_{1}\right);t\in\hat{\mathcal{S}}_{T}\right) & =\bm{0} \label{eq_kkt4}\\
\max_{i}\Vert\frac{1}{NT}\mathbb{X}_{2(i)}^{\top}\left(\bm{y}-\bm{\mu}(\hat{\bm{\theta}})\right)\Vert & <\mathcal{P}_{\gamma_{2}}^{\prime}\left(0+\right)=\lambda_{2} \label{eq_kkt5} \\
\Vert\frac{1}{NT}\mathbb{M}_{2}^{\top}\left(\bm{y}-\bm{\mu}(\hat{\bm{\theta}})\right)\Vert_{\infty} & <\mathcal{P}_{\gamma_{1}}^{\prime}\left(0+\right)=\lambda_{1} \label{eq_kkt6}\\
\lambda_{\min}\left(\frac{1}{NT}\bm{Z}^{\top}\bm{\Sigma}(\hat{\bm{\theta}})\bm{Z}\right) & >0 \label{eq_kkt7}\\
\lambda_{\min}\left(\frac{1}{NT}\bm{X}^{*\top}\bm{\Sigma}(\hat{\bm{\theta}})\bm{X}^{*\top}\right) & >0 \label{eq_kkt8}\\
\lambda_{\min}\left(\frac{1}{NT}\mathbb{X}_{1}^{\top}\bm{\Sigma}(\hat{\bm{\theta}})\mathbb{X}_{1}\right) & > \kappa_{\gamma_2}^{\prime}(\hat{\bm{\varsigma}}_{(1)}) \label{eq_kkt9}\\
\lambda_{\min}\left(\frac{1}{NT}\mathbb{M}_{1}^{\top}\bm{\Sigma}(\hat{\bm{\theta}})\mathbb{M}_{1}\right) & > \kappa_{\gamma_1}(\hat{\bm{\xi}}_{(1)}) \label{eq_kkt10}
\end{align}
where $\mathbb{X}_{2(i)}$ corresponds to the submatrix of $\mathbb{X}_2$ formed by $\hat{\bm{\varsigma}}_i$ for $i\in \hat{\mathcal{S}}_N$, and $\hat{\bm{\theta}}_1$ is the nonzero subvector of $\hat{\bm{\theta}}$.
\end{lem}

\begin{proof}

First, we write the likelihood function  in \eqref{Eq:nll} as $l(\bm{\theta})$ in the following way
\begin{equation}
\label{eq:nll_new}
 l\left(\bm{\theta}\right)=\frac{1}{NT}\left(-\bm{y}^{\top}\left(\bm{Z}\bm{\alpha}+\mathbb{X}\bm{\varsigma}+\mathbb{M}\bm{\xi}\right)+\bm{1}^{\top}\bm{\mu}\left(\bm{\theta}\right)\right)=\frac{1}{NT}\left(-\bm{y}^{\top}\mathbb{U}\bm{\theta}+\bm{1}^{\top}\bm{\mu}\left(\bm{\theta}\right)\right).
\end{equation}
Then, the first order derivative and the second order derivative have the following forms,
\begin{align*}
\nabla l\left(\bm{\theta}\right) & =-\frac{1}{NT}\mathbb{U}^{\top}\left(\bm{y}-\bm{\mu}\left(\bm{\theta}\right)\right),\\
\nabla^{2}l\left(\bm{\theta}\right) & =\frac{1}{NT}\mathbb{U}^{\top}\bm{\Sigma}\left(\bm{\theta}\right)\mathbb{U}.
\end{align*}

Recall $\mathbb{U}=\left(\bm{Z},\mathbb{X},\mathbb{M}\right) = (\bm{Z},\bm{X}^*, \mathbb{X}_1, \mathbb{X}_2, \mathbb{M}_1, \mathbb{M}_2)$, $\hat{\bm{\theta}}$
is a local minimizer of the objective function $Q\left(\bm{\theta}\right)$
if it satisfies the Karush-Kuhn-Tucker (KKT) conditions, that is 
\[
\frac{1}{NT}\mathbb{U}^{\top}\left(\bm{y}-\bm{\mu}\left(\hat{\bm{\theta}}\right)\right)-\bm{v}=\bm{0}.
\]
For $\bm{\alpha}$ and $\bm{\varsigma}_{1}$, there are no penalty
terms so the corresponding $v_{j}=0$. Thus, we have conditions \eqref{eq_kkt1} and \eqref{eq_kkt2}. For $i\in \hat{\mathcal{S}}_N$,
$\Vert\hat{\bm{\varsigma}}_{i}\Vert\neq0$, the corresponding values
in $\bm{v}$ is $\frac{\hat{\bm{\varsigma}}_{i}}{\Vert\hat{\bm{\varsigma}}_{i}\Vert}\mathcal{\mathcal{P}}_{\gamma_{2}}^{\prime}\left(\Vert\hat{\bm{\varsigma}}_{i}\Vert,\lambda_{2}\right)$; for $t\in \hat{\mathcal{S}}_T$, the corresponding element in $\bm{v}$ is $\mathcal{\mathcal{P}}_{\gamma_{1}}^{\prime}\left(\xi_{t},\lambda_{1}\right)$. Thus, we have conditions \eqref{eq_kkt3} and \eqref{eq_kkt4}. For $i \in \hat{\mathcal{N}}_N$, the elements in $\bm{v}$ are $\Vert\bm{v}_{i}\Vert$, which
is any value between $[-\mathcal{P}^{\prime}\left(0+\right),\mathcal{\mathcal{P}}^{\prime}\left(0+\right)]=\left[-\lambda_2,\lambda_2\right]$ \citep{Loh2015},
and for $t\in \hat{\mathcal{N}}_T$, the element in $\bm{v}$ is between $[-\lambda_1,\lambda_1].$  Thus, we have the conditions in \eqref{eq_kkt5} and \eqref{eq_kkt6}.

Now consider the second order conditions. For the parameters without penalty terms, we have \eqref{eq_kkt7} and \eqref{eq_kkt8}. For $\hat{\bm{\varsigma}}_{(1)}$ and $\hat{\bm{\xi}}_{(1)}$, we have 
\begin{align}
\lambda_{\min}\left(\frac{1}{NT}\mathbb{X}_{1}^{\top}\bm{\Sigma}\left(\hat{\bm{\theta}}\right)\mathbb{X}_{1}\right) & \geq\kappa_{\gamma_{2}}^{\prime}(\hat{\bm{\varsigma}}_{(1)}) \label{eq_kappa2},\\
\lambda_{\min}\left(\frac{1}{NT}\mathbb{M}_{1}^{\top}\bm{\Sigma}\left(\hat{\bm{\theta}}\right)\mathbb{M}_{1}\right) & \geq\kappa_{\gamma_{1}}(\hat{\bm{\xi}}_{(1)}) \label{eq_kappa1}.
\end{align}

Note that \eqref{eq_kappa2} and \eqref{eq_kappa1} are nonstrict inequalities. Next, we will show that strict inequalities are sufficient conditions for strict local minimizer. 

First consider the $Q(\bm{\theta})$ on the $\Vert\hat{\bm{\theta}}\Vert_{0}$
dimensional subspace $\mathcal{B}=\left\{ \bm{\theta}\in\mathbb{R}^{q+Np+T-1}:\bm{\theta}_{c}=\bm{0}\right\} $,
where $\bm{\theta}_{c}$ is the subvector of $\bm{\theta}$ formed
by components in $\hat{\mathcal{N}}_{N}$ and $\hat{\mathcal{N}}_{T}$.
Based on second order conditions, $Q\left(\bm{\theta}\right)$ is
strictly convex in a ball $\mathcal{N}_{0}$ in the subspace $\mathcal{B}$
centered at $\hat{\bm{\theta}}.$ Thus, $\hat{\bm{\theta}}$ is a
unique minimizer of $Q\left(\bm{\theta}\right)$ in the neighborhood
$\mathcal{N}_{0}$.

Next, we need to show that $\hat{\bm{\theta}}$ is a strict local minimizer
of $Q\left(\bm{\theta}\right)$ on the space $\mathbb{R}^{q+Np+T-1}$.
Consider a sufficiently small ball $\mathcal{N}_{1}$ in $\mathbb{R}^{q+Np+T-1}$
centered at $\hat{\bm{\theta}}$ such that $\mathcal{N}_{1}\cap\mathcal{B}\subset\mathcal{N}_{0}$.
So we need to show that $Q(\hat{\bm{\theta}})<Q\left(\bm{\phi}_{1}\right)$
for any $\bm{\phi}_{1}\in\mathcal{N}_{1}\backslash\mathcal{N}_{0}$
. Let $\bm{\phi}_{2}$ be the projection of $\bm{\phi}_{1}$ onto
the subspace $\mathcal{B}$. Then, $\bm{\phi}_{2}\in\mathcal{N}_{0}$,
thus $Q(\hat{\bm{\theta}})<Q\left(\bm{\phi}_{2}\right)$
for $\bm{\phi}_{2}\neq\hat{\bm{\theta}}$ since $\hat{\bm{\phi}}$
is a local strict minimizer in $\mathcal{N}_{0}$. Next, we will show
that $Q\left(\bm{\phi}_{2}\right)<Q\left(\bm{\phi}_{1}\right)$. By
the mean-value theorem, 
\begin{equation}
\label{eq:mean_value}
    Q\left(\bm{\phi}_{1}\right)-Q\left(\bm{\phi}_{2}\right)=\left[Q^{\prime}\left(\bm{\phi}_{0}\right)\right]^{\top}\left(\bm{\phi}_{1}-\bm{\phi}_{2}\right),
\end{equation}
where $\bm{\phi}_{0}$ is a vector between $\bm{\phi}_{1}$ and $\bm{\phi}_{2}$,
that is $\bm{\phi}_{0}=\alpha_{0}\bm{\phi}_{1}+\left(1-\alpha_{0}\right)\bm{\phi}_{2}$
and $\alpha_{0}\in\left(0,1\right)$. Based on the definition of $\bm{\phi}_{1}$
and $\bm{\phi}_{2}$, and that $\bm{\alpha},\bm{\varsigma}_{1},\bm{\varsigma}_{(1)},\bm{\xi}_{(1)}$
are the elements of $\bm{\phi}_{1}$ in $\mathcal{B}$, then $\bm{\phi}_{0}=\left(\bm{\alpha},\bm{\varsigma}_{1},\bm{\varsigma}_{\left(1\right)},\bm{\xi}_{\left(1\right)},\alpha_{0}\bm{\varsigma}_{\left(2\right)},\alpha_{0}\bm{\xi}_{\left(2\right)}\right)$.
Thus, the right hand side of \eqref{eq:mean_value} can be expressed as 
\begin{align}
  & -\left[\frac{1}{NT}\mathbb{U}_{2}^\top\left(\bm{y}-\bm{\mu}(\hat{\bm{\phi}}_{0})\right)\right]^\top\alpha_{0}\left(\bm{\varsigma}_{\left(2\right)}^{\top},\bm{\xi}_{\left(2\right)}^{\top}\right)^{\top} \nonumber\\
+ &\sum_{i\in\hat{\mathcal{N}}_{N}}\alpha_{0}\Vert\bm{\varsigma}_{i}\Vert\mathcal{\mathcal{P}}_{\gamma_{2}}^{\prime}\left(\Vert\alpha_{0}\bm{\varsigma}_{i}\Vert,\lambda_{2}\right)+\sum_{t\in\hat{\mathcal{N}}_{T}}\alpha_{0}\mathcal{\mathcal{P}}_{\gamma_{1}}^{\prime}\left(\vert\alpha_{0}\xi_{t}\vert,\lambda_{1}\right)\vert\xi_{t}\vert.  \label{eq:lemma1_bound}
\end{align}

Based on the conditions, there exists some $\delta>0$ such that for
any $\bm{\theta}$ in a ball centered at $\hat{\bm{\theta}}$ with
radius $\delta$, 
\begin{align*}
\Vert\frac{1}{NT}\mathbb{X}_{2(i)}^{\top}\left(\bm{y}-\bm{\mu}\left(\bm{\theta}\right)\right)\Vert_{\infty} & <\mathcal{P}_{\gamma_{2}}^{\prime}\left(\delta,\lambda_{2}\right)=\lambda_{2}-\frac{\delta}{\gamma_{2}},\\
\Vert\frac{1}{NT}\mathbb{M}_{2}^{\top}\left(\bm{y}-\bm{\mu}\left(\bm{\theta}\right)\right)\Vert_{\infty} & <\mathcal{P}_{\gamma_{1}}^{\prime}\left(\delta,\lambda_{1}\right)=\lambda_{1}-\frac{\delta}{\gamma_{1}}.
\end{align*}

Consider the radius of the ball $\mathcal{N}_{1}$ less than $\delta$,
then $\vert\xi_{t}\vert<\delta$ for $t\in\hat{\mathcal{N}}_{T}$ and $\Vert\bm{\varsigma}_{i}\Vert<\delta$ for
$i\in\mathcal{\hat{N}}_{N}$, and above holds for any $\bm{\theta}\in\mathcal{N}_{1}$. Thus, $\mathcal{\mathcal{P}}_{\gamma_{2}}^{\prime}\left(\Vert\alpha_{0}\bm{\varsigma}_{i}\Vert,\lambda_{2}\right) > \mathcal{\mathcal{P}}_{\gamma_{2}}^{\prime}\left(\delta,\lambda_{2}\right)$ and $\mathcal{\mathcal{P}}_{\gamma_{1}}^{\prime}\left(\vert\alpha_{0}\xi_{t}\vert,\lambda_{1}\right) > \mathcal{\mathcal{P}}_{\gamma_{1}}^{\prime}\left(\delta,\lambda_{1}\right)$. Together with $\bm{\phi}_{0}\in\mathcal{N}_{1},$ \eqref{eq:lemma1_bound} is strictly greater than the following
\begin{align*}
 & -\alpha_{0}\left(\lambda_{2}-\frac{\delta}{\gamma_{2}}\right)\sum_{i\in\hat{\mathcal{N}}_{N}}\Vert\bm{\varsigma}_{i}\Vert-\alpha_{0}\left(\lambda_{1}-\frac{\delta}{\gamma_{1}}\right)\sum_{t\in\mathcal{\hat{N}}_{T}}\vert\xi_{t}\vert\\
 & +\alpha_{0}\left(\lambda_{2}-\frac{\delta}{\gamma_{2}}\right)\sum_{i\in\mathcal{\hat{N}}_{N}}\Vert\bm{\varsigma}_{i}\Vert+\alpha_{0}\left(\lambda_{1}-\frac{\delta}{\gamma_{1}}\right)\sum_{t\in\mathcal{\hat{N}}_{T}}\vert\xi_{t}\vert=0.
\end{align*}
Thus, $Q\left(\bm{\phi}_{1}\right)>Q\left(\bm{\phi}_{2}\right)>Q(\hat{\bm{\theta}})$.
This completes the proof. 
\end{proof}

\subsection{Lemma 2}
\label{subsec:lemma2}

\begin{lem}
    \label{lem2} 
    Under the conditions in Theorem 1, given the MST based on the weights in \eqref{Eq:weight}, consider any two locations i and $i^{\prime}$ in the same cluster. Under Condition (C4), there exists a path in the MST connecting $i$ and $i^{\prime}$ such that all the locations on the path belong to the same cluster with probability approaching 1 as local sample size $T\rightarrow\infty$. 
\end{lem}

\begin{proof}
    For a given tree $\mathcal{T}$,  we have $\Vert\hat{\bm{\theta}}-\bm{\theta}_{\mathcal{T},0}\Vert=O_{P}\left(\sqrt{s/N_{0}}\right)$ from Theorem \ref{thm1}.
Recall that $\bm{\varsigma}_{i}$ represents difference between $\bm{\beta}_{i}-\bm{\beta}_{i^{\prime}}$,
$\left(i,i^{\prime}\right)\in\mathcal{E}$. And the $\bm{\beta}=\left(\tilde{\bm{H}}^{-1}\otimes\bm{I}_{p}\right)\bm{\varsigma}$,
then $\Vert\hat{\bm{\beta}}_{i}-\bm{\beta}_{i}^{0}\Vert=O_{P}\left(\sqrt{s/T}\right)$.
Thus, we can calculate the difference between $\Vert\hat{\bm{\beta}}_{i}-\hat{\bm{\beta}}_{i^{\prime}}\Vert$
in $\mathcal{E}_0$. Recall the weights in \eqref{Eq:weight},
\[
w_{ii^{\prime}}=\begin{cases}
\Vert\hat{\bm{\beta}}_{\mathcal{T},i}-\hat{\bm{\beta}}_{\mathcal{T},i^{\prime}}\Vert & \left(i,i^{\prime}\right)\in\mathcal{E}_0,\\
\infty & \text{otherwise}.
\end{cases}
\]

For $(i,i^\prime) \in \mathcal{E}_0$, we have
\begin{align}
\Vert\hat{\bm{\beta}}_{\mathcal{T},i}-\hat{\bm{\beta}}_{\mathcal{T},i^{\prime}}\Vert & =\Vert\hat{\bm{\beta}}_{\mathcal{T},i}- \bm{\beta}_i^0 -(\hat{\bm{\beta}}_{\mathcal{T},i^{\prime}} -\bm{\beta}_{i^\prime}^0) + \bm{\beta}_i^0 - \bm{\beta}_{i^\prime}^0 \Vert \label{eq:weight_bound}
\end{align}


If $i$ and $i^{\prime}$ are in the same cluster, that means $\bm{\beta}_{i}^{0}=\bm{\beta}_{i^{\prime}}^{0}$, then \eqref{eq:weight_bound} becomes
\[
\Vert\hat{\bm{\beta}}_{\mathcal{T},i}-\hat{\bm{\beta}}_{\mathcal{T},i^{\prime}}\Vert \leq \Vert\hat{\bm{\beta}}_{\mathcal{T},i}-\bm{\beta}_{i}^{0}\Vert+\Vert\hat{\bm{\beta}}_{\mathcal{T},i^{\prime}}-\bm{\beta}_{i^{\prime}}^{0}\Vert=O_{P}\left(\sqrt{s/T}\right).
\]

This implies that $\Vert\hat{\bm{\beta}}_{\mathcal{T},i}-\hat{\bm{\beta}}_{\mathcal{T},i^{\prime}}\Vert$ will converge to 0 with probability approaching 1 if $i$ and $i^\prime$ are in the same cluster.

If $i$ and $i^{\prime}$ are in different clusters, then we have, 
\begin{align*}
\Vert\hat{\bm{\beta}}_{\mathcal{T},i}-\hat{\bm{\beta}}_{\mathcal{T},i^{\prime}}\Vert  &\geq  \Vert \bm{\beta}_i^0 - \bm{\beta}_{i^\prime}^0\Vert - \Vert\hat{\bm{\beta}}_{\mathcal{T},i}-\bm{\beta}_{i}^{0}\Vert-\Vert\hat{\bm{\beta}}_{\mathcal{T},i^{\prime}}-\bm{\beta}_{i^{\prime}}^{0}\Vert\\
& = \Vert\bm{\beta}_{i}^{0}-\bm{\beta}_{i^{\prime}}^{0}\Vert - O_{P}\left(\sqrt{s/T}\right).
\end{align*}
This implies that  $\Vert\hat{\bm{\beta}}_{\mathcal{T},i}-\hat{\bm{\beta}}_{\mathcal{T},i^{\prime}}\Vert$ will converge to a constant depending on the cluster difference with probability approaching 1 if $i$ and $i^\prime$ are in different clusters.

Then, by the same arguments of \cite{zhang2019distributed}, the result in Lemma 2 holds.
\end{proof}

Lemma 2 indicates that the graph will be separated into $K$ subgraphs, corresponding to $K$ clusters, by removing edges among clusters.


\section{Theorems}
\label{sec:thms}
\subsection{Proof of Theorem 1}
\label{subsec:thm1}

In this proof, we use $\bm{\theta}_{1,0}$ instead of $\bm{\theta}_{\mathcal{T},1,0}$ for simplicity. And $\bm{\theta}_0 = (\bm{\theta}_{1,0}^\top, \bm{\theta}_{2,0}^\top)^\top$.

\begin{proof}
    Without loss of generality, assume there are $s_{1}$ $\xi_{t}$ 's
are nonzero, they are $\xi_{2},\dots,\xi_{s_{1}+1}$ ,and there are
$s_{2}$ $\Vert\bm{\varsigma}_{i}\Vert$ are nonzero, they are $\Vert\bm{\varsigma}_{2}\Vert,\dots,\Vert\bm{\varsigma}_{s_{2}+1}\Vert$. And $s=s_{1}+s_{2}p$ is the number of nonzero parameters. We will prove the result in two steps.

\paragraph{Step 1: Consistency in the $s$-dimensional subspace. \newline}

Consider the objective function $Q\left(\bm{\theta}\right)$ on the
$s$-dimensional subspace $\left\{ \bm{\theta}\in\mathbb{R}^{q+Np+T-1}:\bm{\theta}_{\mathcal{S}_{0}^{c}}=\bm{0}\right\} $,
where $\mathcal{S}_{0}=\text{supp}\left(\bm{\theta}_{0}\right)$, which is the nonzero set of the parameters,  and $\mathcal{S}_{0}^{c}$ is the complement. Then, the constrained objective
function is 
\[
\bar{Q}\left(\bm{\delta}\right)=\bar{l}\left(\bm{\delta}\right)+\sum_{t=2}^{s_{1}+1}\mathcal{P}_{\gamma_{1}}\left(\vert\xi_{t}\vert,\lambda_{1}\right)+\sum_{i=2}^{s_{2}+1}\mathcal{P}_{\gamma_{2}}\left(\Vert\bm{\varsigma}_{i}\Vert,\lambda_{2}\right),
\]
where $\bar{l}\left(\bm{\delta}\right)=\frac{1}{NT}\left(-\bm{y}^{\top}\mathbb{U}_{1}\bm{\delta}+\bm{1}^{\top}\bm{\mu}\left(\bm{\delta}\right)\right)=\frac{1}{NT}\left(-\bm{y}^{\top}\left(\bm{Z}\bm{\alpha}+\mathbb{X}_{1}\bm{\varsigma}_{\left(1\right)}+\mathbb{M}_{1}\bm{\xi}_{\left(1\right)}\right)+\bm{1}^{\top}\bm{\mu}\left(\bm{\delta}\right)\right)$,
here $\bm{\mu}\left(\bm{\delta}\right)=\left(n_{it}\exp\left(b_{it}\right)\right)_{1\le i\leq N,1\leq t\leq T}$,
with $\bm{b}=\mathbb{U}_{1}\bm{\delta}$ . We now show that there
exists a strict local minimizer of $\hat{\bm{\theta}}_{1}$ of $\bar{Q}\left(\bm{\delta}\right)$
such that $\Vert\hat{\bm{\theta}}_{1}-\bm{\theta}_{1,0}\Vert=O_{P}\left(\sqrt{s/N_{0}}\right)$,
where $N_{0}=NT$.

Define an event
\[
H_{1}=\left\{ \bar{Q}\left(\bm{\theta}_{1,0}\right)<\min_{\bm{\delta}\in\partial\mathcal{N}_{C}}\bar{Q}\left(\bm{\delta}\right)\right\},
\]
where $\partial\mathcal{N}_{C}$ denotes the boundary of the closet
set $\mathcal{N}_{C}=\left\{ \Vert\bm{\delta}-\bm{\theta}_{1,0}\Vert\leq C\sqrt{s/N_{0}}\right\} $
and $C\in\left(0,\infty\right)$. On event $H_{1}$, there exists
a local minimizer $\hat{\bm{\theta}}_{1}$ of $\bar{Q}\left(\bm{\delta}\right)$
in $\mathcal{N}_{C}$. We need to show that $P\left(H_{1}\right)$
is close to 1 as $T\rightarrow\infty$ when $C$ is large. 

Let $T$ be sufficiently large such that $\sqrt{s/N_{0}}C\leq d$ by Condition (C3),
$\delta_{\xi_{t}}$ is the element in $\bm{\delta}$ in $\mathcal{N}_{C}$
corresponds to $\xi_{t}$, and $\bm{\delta}_{\bm{\varsigma}_{i}}$
is the element in $\bm{\delta}$ in $\mathcal{N}_{C}$ corresponds
to $\bm{\varsigma}_{i}$. By Taylor expansion, for $\bm{\delta}\in\partial\mathcal{N}_{C}$,
\[
\bar{Q}\left(\bm{\delta}\right)-\bar{Q}\left(\bm{\theta}_{1,0}\right)=\left(\bm{\delta}-\bm{\theta}_{1,0}\right)^{\top}\bm{v}+\frac{1}{2}\left(\bm{\delta}-\bm{\theta}_{1,0}\right)^{\top}\bm{D}\left(\bm{\delta}-\bm{\theta}_{1,0}\right),
\]
where $\bm{v}$ is the first order derivative of $\bar{Q}\left(\bm{\delta}\right)$
evaluated at $\bm{\theta}_{1,0}$, and $\bm{D}$ is the second order
derivative of $\bar{Q}\left(\bm{\delta}\right)$ evaluated at $\bm{\theta}_{1}^{*}$
with $\bm{\theta}_{1}^{*}=\alpha_{1}\bm{\delta}+\left(1-\alpha_{1}\right)\bm{\theta}_{1,0}$
and $\alpha_{1}\in\left(0,1\right)$. We have 
\begin{align*}
 \bm{v}= & -\frac{1}{NT}\mathbb{U}_{1}^{\top}\left(\bm{y}-\bm{\mu}\left(\bm{\theta}_{1,0}\right)\right) \\
 +&\left(\bm{0},\frac{\bm{\varsigma}_{i}}{\Vert\bm{\varsigma}_{i}\Vert}\mathcal{P}_{\gamma_{2}}^{\prime}\left(\Vert\bm{\varsigma}_{i}\Vert,\lambda_{2}\right);i=2,\dots,s_{2}+1,\mathcal{P}_{\gamma_{1}}^{\prime}\left(\xi_{t},\lambda_{1}\right);t=2,\dots,s_{1}+1\right), 
 \end{align*}

and 
\begin{align*}
\bm{D}= & \frac{1}{NT}\mathbb{U}_{1}^{\top}\bm{\Sigma}\left(\bm{\theta}_{1}^{*}\right)\mathbb{U}_{1}\\
+ & \text{\text{diag}}\left(\bm{0},\left[\frac{\bm{\varsigma}_{i}^{*}}{\Vert\bm{\varsigma}_{i}^{*}\Vert}\mathcal{P}_{\gamma_{2}}^{\prime}\left(\Vert\bm{\varsigma}_{i}^{*}\Vert,\lambda_{2}\right)\right]^{\prime};i=2,\dots,s_{2}+1,\mathcal{P}_{\gamma_{1}}^{\prime\prime}\left(\xi_{t}^{*},\lambda_{1}\right);t=2,\dots,s_{1}+1\right).
\end{align*}

Recall that $\bm{\delta}\in\partial\mathcal{N}_{C}$, thus $\Vert\bm{\delta}-\bm{\theta}_{1,0}\Vert=C\sqrt{s/N_{0}}$. Since $\bm{\theta}_{1}^{*}=\alpha_{1}\bm{\delta}+\left(1-\alpha_{1}\right)\bm{\theta}_{1,0}$, we have $\Vert\bm{\theta}_{1}^{*}-\bm{\theta}_{1,0}\Vert=\Vert\alpha_{1}\left(\bm{\delta}-\bm{\theta}_{1,0}\right)\Vert=\alpha_{1}\Vert\bm{\delta}-\bm{\theta}_{1,0}\Vert,$ thus
$\bm{\theta}_{1}^{*}\in\mathcal{N}_{0}$.  When $T$ is sufficiently
large, since $d\gg\sqrt{s/N_{0}},$, then, $\vert\xi_{t}^{*}\vert\geq\vert\xi_{t}\vert-C\sqrt{s/N_{0}}\geq d$
, and $\Vert\bm{\varsigma}_{i}^{*}\Vert\geq\Vert\bm{\varsigma}_{i}\Vert-C\sqrt{s/N_{0}}\geq d$ based on Condition (C3).
Based on the definition of MCP and assumption that $d\gg\max\left(\lambda_{1},\lambda_{2}\right)$,
we have $d\gg\max\left(\gamma_1\lambda_{1},\gamma_2 \lambda_{2}\right)$, thus
$\mathcal{P}_{\gamma_{1}}^{\prime\prime}\left(\xi_{t}^{*},\lambda_{1}\right)=0$
and $\mathcal{P}_{\gamma_{2}}^{\prime\prime}\left(\Vert\bm{\varsigma}_{i}^{*}\Vert,\lambda_{2}\right)=0$
and $\mathcal{P}_{\gamma_{2}}^{\prime}\left(\Vert\bm{\varsigma}_{i}^{*}\Vert,\lambda_{2}\right)=0$. Then, we have the following result: 
\begin{align*}
\left[\frac{\bm{\varsigma}_{i}^{*}}{\Vert\bm{\varsigma}_{i}^{*}\Vert}\mathcal{P}_{\gamma_{2}}^{\prime}\left(\Vert\bm{\varsigma}_{i}^{*}\Vert,\lambda_{2}\right)\right]^{\prime} & =\frac{\bm{\varsigma}_{i}^{*}\bm{\varsigma}_{i}^{*\top}}{\Vert\bm{\varsigma}_{i}^{*}\Vert^{2}}\mathcal{P}_{\gamma_{2}}^{\prime\prime}\left(\Vert\bm{\varsigma}_{i}^{*}\Vert,\lambda_{2}\right)+\frac{\mathcal{P}_{\gamma_{2}}^{\prime}\left(\Vert\bm{\varsigma}_{i}^{*}\Vert,\lambda_{2}\right)}{\Vert\bm{\varsigma}_{i}^{*}\Vert}\bm{I}_{p}-\frac{\bm{\varsigma}_{i}^{*}\bm{\varsigma}_{i}^{*\top}\mathcal{P}_{\gamma_{2}}^{\prime}\left(\Vert\bm{\varsigma}_{i}^{*}\Vert,\lambda_{2}\right)}{\Vert\bm{\varsigma}_{i}^{*}\Vert^{3}}=\bm{0}.
\end{align*}
Thus, $\lambda_{\min}\left(\bm{D}\right)=\lambda_{\min}\left(\frac{1}{NT}\mathbb{U}_{1}^{\top}\bm{\Sigma}\left(\bm{\theta}_{1}^{*}\right)\mathbb{U}_{1}\right)\geq c_{1}M_{1}$ based on Condition (C1) and (C2). Also within $\mathcal{N}_{C}$, we have $\vert\delta_{\xi_{t}}\vert\geq\vert\xi_{t}\vert-C\sqrt{s/N_{0}}\geq d$,
and $\Vert\bm{\delta}_{\bm{\varsigma}_{i}}\Vert\geq\Vert\bm{\varsigma}_{i}\Vert-C\sqrt{s/N_{0}}\geq d$, 
$\Vert\bm{\varsigma}_{i}\Vert\geq d$ and $\vert\xi_{t}\vert\geq d$,
thus $\mathcal{P}_{\gamma_{2}}^{\prime}\left(\Vert\bm{\varsigma}_{i}\Vert,\lambda_{2}\right)=0$
and $\mathcal{P}_{\gamma_{1}}^{\prime}\left(\xi_{t},\lambda_{1}\right)=0$
based on $d\gg\max\left(\gamma_{1}\lambda_{1},\gamma_{2}\lambda_{2}\right)$. Thus, we have
\begin{align*}
 & \min_{\bm{\delta}\in\partial\mathcal{N}_{C}}\bar{Q}\left(\bm{\delta}\right)-\bar{Q}\left(\bm{\theta}_{1,0}\right)\\
\geq & -\Vert\frac{1}{NT}\mathbb{U}_{1}^{\top}\left(\bm{y}-\bm{\mu}\left(\bm{\theta}_{1,0}\right)\right)\Vert\Vert\bm{\delta}-\bm{\theta}_{1,0}\Vert + c_1M_1 \frac{1}{2} \Vert \bm{\delta} - \bm{\theta}_{1,0}\Vert^2
\\
= & -C\sqrt{s/N_{0}}\Vert\frac{1}{NT}\mathbb{U}_{1}^{\top}\left(\bm{y}-\bm{\mu}\left(\bm{\theta}_{1,0}\right)\right)\Vert+c_{1}M_{1}\frac{1}{2}\Vert\bm{\delta}-\bm{\theta}_{1,0}\Vert^{2}\\
= & -C\sqrt{s/N_{0}}\left[\Vert\frac{1}{NT}\mathbb{U}_{1}^{\top}\left(\bm{y}-\bm{\mu}\left(\bm{\theta}_{1,0}\right)\right)\Vert-\frac{c_{1}M_{1}C\sqrt{s/N_{0}}}{2}\right].
\end{align*}

Based on Markov's inequality, 
\begin{align*}
P\left(H_{1}\right) & \geq P\left[\Vert\frac{1}{NT}\mathbb{U}_{1}^{\top}\left(\bm{y}-\bm{\mu}\left(\bm{\theta}_{1,0}\right)\right)\Vert^{2}<\frac{c_{1}^{2}M_{1}^{2}C^{2}s}{4N_{0}}\right]\\
 & \geq1-\frac{4N_{0}E\Vert\frac{1}{NT}\mathbb{U}_{1}^{\top}\left(\bm{y}-\bm{\mu}\left(\bm{\theta}_{1,0}\right)\right)\Vert^{2}}{c_{1}^{2}M_{1}^{2}C^{2}s}.
\end{align*}

We know that $E\left(\bm{y}\right)=\bm{\mu}\left(\bm{\theta}_{0}\right)$, $Var\left(\bm{y}\right)=\bm{\Sigma}\left(\bm{\theta}_{0}\right)=\text{diag}\left(\bm{\mu}\left(\bm{\theta}_{1,0}\right)\right)$, and \\
$E\Vert\frac{1}{NT}\mathbb{U}_{1}^{\top}\left(\bm{y}-\bm{\mu}\left(\bm{\theta}_{1,0}\right)\right)\Vert^{2}=\frac{1}{N_{0}^{2}}\text{tr}\left(\mathbb{U}_{1}^{\top}\bm{\Sigma}\left(\bm{\theta}_{0}\right)\mathbb{U}_{1}\right)\leq\frac{sc_{2}M_{2}}{N_{0}}$ by Condition (C1) and (C2). Thus
\[
P\left(H_{1}\right)\geq1-\frac{4N_{0}sc_{2}M_{2}}{N_{0}c_{1}^{2}M_{1}^{2}C^{2}s}=1-\frac{4c_{2}M_{2}}{c_{1}^{2}M_{1}^{2}}\frac{1}{C^{2}}.
\]
As $C\rightarrow\infty$, $P\left(H_{1}\right)\rightarrow1$. This proves that $\Vert \hat{\bm{\theta}}_1 - \bm{\theta}_{1,0}\Vert = O_P(\sqrt{s/N_0})$.

\paragraph{Step 2: Sparsity. \newline}

Let $\hat{\bm{\theta}}=\left(\hat{\bm{\theta}}_{1},\hat{\bm{\theta}}_{2}\right)$,
where $\hat{\bm{\theta}}_{1}\in\mathcal{N}_{C}\subset\mathcal{N}_{0}$
is a strict local minimzer of $\bar{Q}\left(\bm{\delta}\right)$ and
$\hat{\bm{\theta}}_{2}=\bm{0}$. We need to show that $\hat{\bm{\theta}}$
is a strict local minimzer of $Q\left(\bm{\theta}\right)$.

Based on conditions (C1) and (C2) and the definition of $\hat{\bm{\theta}}_1$, conditions in \eqref{eq_kkt1}, \eqref{eq_kkt2}, \eqref{eq_kkt3}, \eqref{eq_kkt4} \eqref{eq_kkt7}, \eqref{eq_kkt8}, \eqref{eq_kkt9} and \eqref{eq_kkt10} are satisfied. Thus, it suffices to check conditions \eqref{eq_kkt5} and \eqref{eq_kkt6}.

Let $\bm{\varphi}=\mathbb{U}^{\top}\left(\bm{y}-\bm{\mu}\left(\bm{\theta}_{0}\right)\right)$,
and consider the event 
\[
\mathcal{E}_{1}=\left\{ \max_{i=s_{2}+2,\dots N}\Vert\mathbb{X}_{2\left(i\right)}\left(\bm{y}-\bm{\mu}\left(\bm{\theta}_{0}\right)\right)\Vert\leq u\sqrt{N_{0}};\quad\sup_{t=s_{1}+2,\dots,T}\vert\mathbb{M}_{2(t-1)}\left(\bm{y}-\bm{\mu}\left(\bm{\theta}_{0}\right)\right)\vert\leq u\sqrt{N_{0}}\right\} .
\]

From \cite{fan2011nonconcave}, we have 
\begin{equation}
  P\left(\vert\bm{a}^{\top}\left(\bm{y}-\bm{\mu}\left(\bm{\theta}_{0}\right)\right)\vert>\Vert\bm{a}\Vert\epsilon\right)\leq2e^{-c_{3}\epsilon^{2}},  
  \label{eq_prob}
\end{equation}
where $c_{3}$ is a positive constant. Let $\mathbb{M}_{\left(j\right)}$
be the $j$th column of $\mathbb{M}$ and satisfy $\max_j\Vert\mathbb{M}_{\left(j\right)}\Vert=\sqrt{\left(T-1\right)N}$
based on the construction of $\mathbb{M}$. Let $\mathbb{X}_{2\left(ij\right)}$
be the $j$th column of $\mathbb{X}_{2\left(i\right)}$ and $\max_{i,j}\Vert\mathbb{X}_{2\left(i,j\right)}\Vert=O\left(\sqrt{N_{0}}\right)$
based on condition (C1). It follows from Bonferroni's
inequality, we have
\begin{align*}
P\left(\mathcal{E}_{1}\right) & \geq1-\sum_{i=s_{2}+2}^{N}\sum_{j=1}^{p}P\left(\Vert\mathbb{X}_{2\left(ij\right)}^{\top}\left(\bm{y}-\bm{\mu}\left(\bm{\theta}_{0}\right)\right)\Vert>u\sqrt{N_{0}}\right) \\
&-\sum_{t=s_{1}+2}^{T}P\left(\vert\mathbb{M}_{2(t-1)}^{\top}\left(\bm{y}-\bm{\mu}\left(\bm{\theta}_{0}\right)\right)\vert>u\sqrt{N_{0}}\right),
\end{align*}
where $u\gg\sqrt{\log N_{0}}$ and $uN_{0}^{-1/2}=o\left(1\right)$. 

From \eqref{eq_prob}, we have
\begin{align*}
P\left(\Vert\mathbb{X}_{2\left(ij\right)}^{\top}\left(\bm{y}-\bm{\mu}\left(\bm{\theta}_{0}\right)\right)\Vert>u\sqrt{N_{0}}\right) & =P\left(\Vert\mathbb{X}_{2\left(ij\right)}^{\top}\left(\bm{y}-\bm{\mu}\left(\bm{\theta}_{0}\right)\right)\Vert>\Vert\mathbb{X}_{2\left(ij\right)}^{\top}\Vert\frac{u\sqrt{N_{0}}}{\Vert\mathbb{X}_{2\left(ij\right)}\Vert}\right)\\
 & \leq P\left(\Vert\mathbb{X}_{2\left(ij\right)}^{\top}\left(\bm{y}-\bm{\mu}\left(\bm{\theta}_{0}\right)\right)\Vert>\Vert\mathbb{X}_{2\left(ij\right)}^{\top}\Vert\frac{u\sqrt{N_{0}}}{\max_{i,j}\Vert\mathbb{X}_{2\left(ij\right)}\Vert}\right)\\
 & \leq2e^{-c_{3}u^{2}}
\end{align*}
and 
\begin{align*}
P\left(\vert\mathbb{M}_{2(t-1)}^{\top}\left(\bm{y}-\bm{\mu}\left(\bm{\theta}_{0}\right)\right)\vert>u\sqrt{N_{0}}\right) & =P\left(\vert\mathbb{M}_{2(t-1)}^{\top}\left(\bm{y}-\bm{\mu}\left(\bm{\theta}_{0}\right)\right)\vert>\Vert\mathbb{M}_{2(t-1)}^{\top}\Vert\frac{u\sqrt{N_{0}}}{\Vert\mathbb{M}_{2(t-1)}\Vert}\right)\\
 & \leq P\left(\vert\mathbb{M}_{2(t-1)}^{\top}\left(\bm{y}-\bm{\mu}\left(\bm{\theta}_{0}\right)\right)\vert>\Vert\mathbb{M}_{2(t-1)}^{\top}\Vert\frac{u\sqrt{N_{0}}}{\max_j\Vert\mathbb{M}_{2(t-1)}\Vert}\right)\\
 & \leq P\left(\vert\mathbb{M}_{2(t-1)}^{\top}\left(\bm{y}-\bm{\mu}\left(\bm{\theta}_{0}\right)\right)\vert>\Vert\mathbb{M}_{2(t-1)}^{\top}\Vert\frac{u\sqrt{N_{0}}}{\sqrt{N_{0}}}\right)\\
 & \leq2e^{-c_{3}u^{2}}.
\end{align*}
Thus 
\[
P\left(\mathcal{E}_{1}\right)\geq1-2\left((N-s_{2}-1)p+T-s_{1}-1\right)e^{-c_{3}u^{2}}\geq1-2\frac{Np+T}{e^{c_{3}u^{2}}}\geq 1-\frac{2}{e^{c_{3}u^{2}-\log N_{0}}}\rightarrow1.
\]

Under event $\mathcal{E}_{1}$, now consider the following in \eqref{eq_kkt5},
\[
\max_{i}\Vert\frac{1}{NT}\mathbb{X}_{2(i)}^{\top}\left(\bm{y}-\bm{\mu}(\hat{\bm{\theta}})\right)\Vert\leq \max_{i}\Vert\frac{1}{NT}\mathbb{X}_{2\left(i\right)}^{\top}\left(\bm{y}-\bm{\mu}\left(\bm{\theta}_{0}\right)\right)\Vert+ \max_{i}\Vert\frac{1}{NT}\mathbb{X}_{2(i)}^{\top}\left(\bm{\mu}\left(\bm{\theta}_{0}\right)-\bm{\mu}(\hat{\bm{\theta}})\right)\Vert.
\]

The first part is bounded by $\frac{u\sqrt{N_{0}}}{N_{0}}=uN_{0}^{-1/2}\ll\lambda_{2}$ by condition (C3) under $\mathcal{E}_{1}$. 

Next, consider the second part $\frac{1}{NT}\mathbb{X}_{2(i)}^{\top}\left(\bm{\mu}(\hat{\bm{\theta}})-\bm{\mu}\left(\bm{\theta}_{0}\right)\right)$.
According to the Taylor expansion,
\begin{align*}
\mathbb{X}_{2(i)}^{\top}\left(\bm{\mu}(\hat{\bm{\theta}})-\bm{\mu}\left(\bm{\theta}_{0}\right)\right) & =\mathbb{X}_{2\left(i\right)}^{\top}\left(\bm{\mu}(\hat{\bm{\theta}}_{1})-\bm{\mu}\left(\bm{\theta}_{1,0}\right)\right)\\
 & =\mathbb{X}_{2\left(i\right)}^{\top}\bm{\Sigma}\left(\bm{\theta}_{1,0}\right)\mathbb{X}_{1}\left(\hat{\bm{\theta}}_{1}-\bm{\theta}_{1,0}\right)+\bm{w},
\end{align*}
where $w_{j}=\frac{1}{2}\left(\hat{\bm{\theta}}_{1}-\bm{\theta}_{1,0}\right)^{\top}\nabla_{j}\left(\hat{\bm{\theta}}_{1}-\bm{\theta}_{1,0}\right)$,
$\nabla_{j}=\mathbb{X}_{1}^{\top}\text{\text{diag}}\left(x_{2(ij)}\bm{\mu}\left(\bm{\theta}_{1}^{*}\right)\right)\mathbb{X}_{1}$ and $\bm{\theta}_1^*$ is a vector lying on the line segment jointing $\hat{\bm{\theta}}_1$ and $\bm{\theta}_{1,0}$. 
Since all $x$ are bounded, then $\Vert\bm{w}\Vert_{\infty}\leq\frac{1}{2}M_{x}N_{0}\lambda_{\max}\left[\frac{1}{NT}\mathbb{X}_{1}^{\top}\mathbb{X}_{1}\right]\Vert\hat{\bm{\theta}}_{1}-\bm{\theta}_{1,0}\Vert^{2}=O\left(N_{0}\right)\Vert\hat{\bm{\theta}}_{1}-\bm{\theta}_{1,0}\Vert^{2}$,
 where $M_{x}$ is a positive constant. Similarly,
\[
\Vert\mathbb{X}_{2\left(i\right)}^{\top}\bm{\Sigma}\left(\bm{\theta}_{1,0}\right)\mathbb{X}_{1}\left(\hat{\bm{\theta}}_{1}-\bm{\theta}_{1,0}\right)\Vert=O\left(N_{0}\right)\Vert\hat{\bm{\theta}}_{1}-\bm{\theta}_{1,0}\Vert.
\]
Thus,
\[
\max_i\Vert\frac{1}{NT}\mathbb{X}_{2(i)}^{\top}\left(\bm{y}-\bm{\mu}\left(\hat{\bm{\theta}}\right)\right)\Vert=o\left(\lambda_{2}\right)+O\left(\Vert\hat{\bm{\theta}}_{1}-\bm{\theta}_{1,0}\Vert+\Vert\hat{\bm{\theta}}_{1}-\bm{\theta}_{1,0}\Vert^{2}\right).
\]
Since $\lambda_{2}\gg\sqrt{s/N_{0}},$ thus $\Vert\hat{\bm{\theta}}_{1}-\bm{\theta}_{1,0}\Vert+\Vert\hat{\bm{\theta}}_{1}-\bm{\theta}_{1,0}\Vert^{2}=o\left(\lambda_{2}\right)$.
Thus $\max_i\Vert\frac{1}{NT}\mathbb{X}_{2(i)}^{\top}\left(\bm{y}-\bm{\mu}\left(\hat{\bm{\theta}}\right)\right)\Vert=o\left(\lambda_{2}\right)$.
The condition in \eqref{eq_kkt5} holds. 

By similar arguments, we can have \eqref{eq_kkt6} hold. This completes the proof.

\end{proof}

\subsection{Proof of Theorem 2}
\label{subsec:thm2}
\begin{proof}
Let $\mathcal{W}$ be an event that for a given MST, $\mathcal{T}$, for
any two locations $i$ and $i^{\prime}$ in the same cluster, there
exists a path in the MST connecting $i$ and $i^{\prime}$ such that
all the locations on the path belong to the same cluster. This means
that in event $\mathcal{W}$, when the nonzero values of $\bm{\varsigma}_{i}$ are identified, which means the nonzero edges are identified. This implies that the true cluster structure of $\bm{\beta}$ can be recovered. 

For $C_0 \in (0,\infty)$ , we have 
\begin{align*}
P\left(\Vert\hat{\bm{\theta}}-\bm{\theta}_{0}^{*}\Vert\leq C_{0}\sqrt{s/N_{0}}\right) & \geq P\left((\Vert\hat{\bm{\theta}}-\bm{\theta}_{0}^{*}\Vert\leq C_{0}\sqrt{s/N_{0}})\cap\mathcal{W}\right)\\
 & =P\left(\Vert\hat{\bm{\theta}}-\bm{\theta}_{0}^{*}\Vert\leq C_{0}\sqrt{s/N_{0}}\vert\mathcal{W}\right)P\left(\mathcal{W}\right).
\end{align*}
Since given $\mathcal{W}$, we know that $P\left(\Vert\hat{\bm{\theta}}-\bm{\theta}_{0}^{*}\Vert\leq C_{0}\sqrt{s/N_{0}}\vert\mathcal{W}\right)\rightarrow1$
as $T\rightarrow\infty$ from Theorem 1. And $P\left(\mathcal{W}\right)\rightarrow1$
as $T\rightarrow\infty$ from Lemma 2. Thus $P\left(\Vert\hat{\bm{\theta}}-\bm{\theta}_{0}^{*}\Vert\leq C_{0}\sqrt{s/N_{0}}\right)\rightarrow1$
as $T\rightarrow\infty$. 

Besides this, we also have the following results.
\begin{align*}
P\left(\hat{\mathcal{S}}^*_N=\mathcal{S}_N^*\right) & =P\left(\hat{\mathcal{S}}^*_N=\mathcal{S}_N^*\vert\mathcal{W}\right)P\left(\mathcal{W}\right)+P\left(\hat{\mathcal{S}}^*_N=\mathcal{S}_N\vert\mathcal{W}^{c}\right)P\left(\mathcal{W}^{c}\right)\\
 & \rightarrow 1,
\end{align*}
since $P(\mathcal{W}) \rightarrow 1$ and $P(\mathcal{W}^c) \rightarrow 0$. We can also have $P(\hat{\mathcal{S}}^*_T = \mathcal{S}_T)  \rightarrow 1$ by the same arguments. 
\end{proof}

\section{Scan approaches in real data analysis}

 When no covariate is considered, NPFSS, PSS and PPSS perform better than other approaches among the scan approaches, so we also use NPFSS, PSS and PPSS to analyze the dataset. PSS and PPSS give the same estimated cluster structure. Figure \ref{fig:scan_cluster} shows the estimated cluster structures based on NPFSS and PPS, respectively. They give very different cluster structures. And in our analysis, we consider the effects after adjusting the covariates, which is different from these scan approaches. 

\begin{figure}[H]
\centering
\begin{subfigure}{.45\textwidth}
  \centering
  \includegraphics[width=0.95\linewidth]{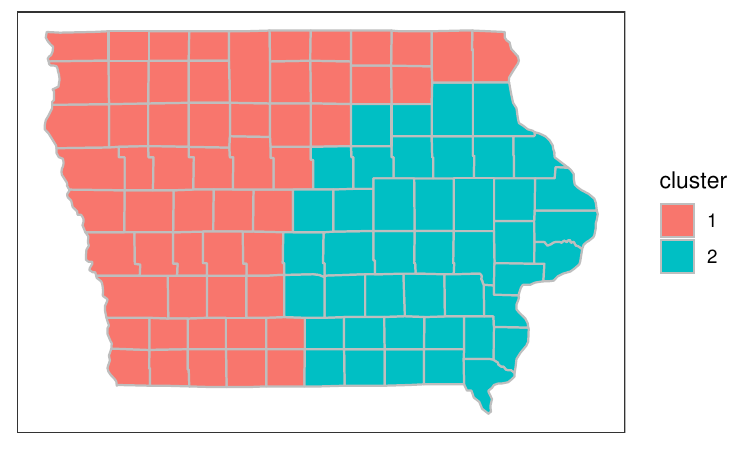}
  \caption{Estimated cluster structure based on NPFSS.}
  \label{fig:npfss}
\end{subfigure}%
\hspace{.1in}
\begin{subfigure}{.45\textwidth}
  \centering
  \includegraphics[width=0.95\linewidth]{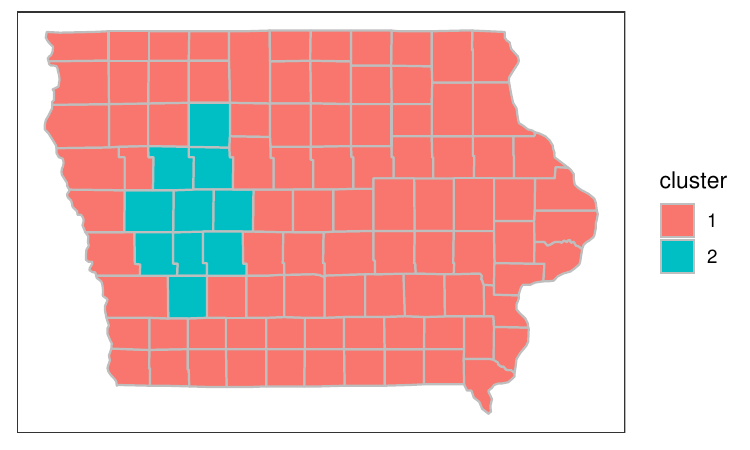}
  \caption{Estimated cluster structure based on PSS and PPSS.}
  \label{fig:pss}
\end{subfigure}
\caption{Estimated cluster structures based on scan approaches.}
\label{fig:scan_cluster}
\end{figure}

\end{appendices}

\bibliographystyle{elsarticle-harv} 
\bibliography{reference}

\end{document}